%% file: rlfhTR.tex
\documentclass[a4paper,UKenglish]{lipics-v2016}
\usepackage{microtype} 

\usepackage{latexsym}
\usepackage{mathpartir} 
\usepackage{color} 
\usepackage{stackrel} 
\usepackage[only,owedge,llfloor,rrfloor,lightning,lbag,rbag,Longmapsto]{stmaryrd} 

\usepackage{tikz} 
\usetikzlibrary{arrows,decorations.pathmorphing,backgrounds,positioning,fit,calc}

\title{Relational logic with framing and hypotheses:
technical report\footnote{Banerjee was partially supported by the US National Science Foundation (NSF). Naumann and Nikouei
were partially supported by NSF award 1228930. Any opinion, findings, and conclusions or recommendations expressed in the material are those of the authors and do not necessarily reflect the views of NSF.}
}

\titlerunning{Relational logic with framing and hypotheses (Technical Report)} 

\author[1]{Anindya Banerjee}
\author[2]{David A. Naumann}
\author[2]{Mohammad Nikouei}
\affil[1]{IMDEA Software Institute}
\affil[2]{Stevens Institute of Technology}

\authorrunning{A.\,Banerjee, D.\,A.\,Naumann, and M.\ Nikouei}

\Copyright{Anindya Banerjee, David A. Naumann, Mohammad Nikouei}

\input{lipmacros}

\newcommand{\fullstuff}[1]{{\color{blue}\protect#1}} 

\lstdefinelanguage{RL}[ANSI]{C}{%
  morekeywords={requires,ensures,modifies,reads,writes,
    rgn,procedure,module,end,
    assert,assume,bool,true,false,skip,od,fi,
    result,forall,exists, then,class,method,foreach,var,
    specpublic, ghost, pure, self, private, null
},
  columns = fullflexible,
  morecomment=[is]{//--}{//--},
   basicstyle=\rmfamily,
   identifierstyle=\textit,
  breaklines=true,
  literate= 
    {\\eqdef}{{$\quad\eqdef\quad$}}1
    {\\cat}{{$\cat$}}1
    {\\geq}{{$\geq$}}1
    {\\imp}{{$\imp\:$}}1
    {\\neq}{{$\neq$}}1
    {\\Agr}{{$\Agr\:$}}1
    {\\Img}{{$\Img\!\!$}}1
    {\\union}{{$\union\:$}}1
    {\\in}{{$\in\:$}}1
    {\\notin}{{$\notin\:$}}1
    {\\land}{{$\land\:$}}1
    {\\neg}{{$\neg$}}1 
    {\\lor}{{$\lor$}}1
    {\\forall}{{$\forall{}$}}1
    {\\exists}{{$\exists{}$}}1
    {\\disj}{{$\#\:$}}1
    {\\eqbi}{{$\eqbi\:$}}1
    {\\prime}{{$^\prime$}}1         
}
\begin{document}

\maketitle

\begin{abstract}
Relational properties arise in many settings: relating two versions of a program that use different data representations, noninterference properties for security, etc.
The main ingredient of relational verification, relating aligned pairs of intermediate steps, 
has been used in numerous guises,
but existing relational program logics are narrow in scope.
This paper introduces a logic based on novel syntax
that weaves together product programs to express alignment of control flow points at which relational formulas are asserted.  Correctness judgments feature hypotheses with relational specifications, discharged by a rule for the linking of procedure implementations. The logic supports reasoning about program-pairs containing both similar and dissimilar control and data structures.
Reasoning about dynamically allocated objects is supported by a frame rule based on frame conditions amenable to SMT provers.
We prove soundness and sketch how the logic can be used for 
data abstraction,
loop optimizations, 
and secure information flow.
 \end{abstract}

\section{Introduction}\label{sec:intro}

Relational properties are ubiquitous.  Compiler optimizations, changes of data representation, and refactoring involve two different programs.  Non-interference (secure information flow) is a non-functional property of a single program; it says the program preserves a ``low indistinguishability'' relation~\cite{VolpanoSI96}.
Many recent works deal with one or more of these applications, using relational logic and/or some form of product construction that reduces the problem to partial correctness, 
though mostly for 
simple imperative programs. This paper advances extant work by providing a relational logic for local reasoning about heap data structures and programs with procedures.

To set the stage, first consider the two simple imperative programs:
\labf{C}{
x:=1;\ \WHILE\; y>0\;\DO\;x:=x * y;\ y:=y-1\;\OD
}
\vspace*{-2ex}
\labf{C'}{
\!\!\!
\begin{array}[t]{l}
x:=1;\ y:=y-1; \;
\WHILE\; y\geq 0\;\DO\; x:=x*y+x;\ y:=y-1\;\OD
\end{array}
}
Both $C$ and $C'$ change $x$ to be the factorial of the initial value of $y$, or to 1 if $y$ is initially negative.  
For a context where $y$ is known to be positive and its final value is not used, 
we could reason that they are interchangeable by showing both 
\begin{equation}\label{eq:functCorr} 
C\,:\;\: y = z \land y \geq 0 \;\leadsto\; x=z! \quad \mbox{and} 
\quad
C'\,:\;\: y = z \land y \geq 0 \;\leadsto\; x=z!
\end{equation}
This is our notation for partial correctness judgments, with evident pre- and postconditions, for $C$ and $C'$.  
It is not always easy to express and prove functional correctness, which motivates a
less well developed approach to showing interchangeability of the examples.
The two programs have a relational property which we write as
\begin{equation}\label{eq:factEq}
 \splitbi{C}{C'}\,:\;\: \Both{(y\geq 0)}\land y\eqbi y \;\rspecSym\; x \eqbi x
\end{equation}
This relational correctness judgment says that a pair of terminating executions of $C$ and $C'$, from a pair of states which both satisfy $y\geq 0$ 
and which agree on the value of $y$, yields a pair of final states that agree on the value of $x$.  
The relational formula $x\eqbi x$ says that the value of $x$ in the left state is the same as its value in the right state.

Property (\ref{eq:factEq})
is a consequence of functional correctness (\ref{eq:functCorr}),
but there is a direct  way to prove it.
Any pair of runs, from states that agree on $y$, can be aligned in such a way that both $x\eqbi x$ and $y\eqbi y+1$ hold at the aligned pairs of intermediate states.  The alignment is almost but not quite step by step, owing to the additional assignment in $C'$.
The relational property is more complicated than partial correctness, in that it involves pairs of runs. 
On the other hand the requisite intermediate assertions are much simpler; they do not involve $!$ which is recursively defined. 
Prior work showed 
such assertions are amenable to automated inference (see Section~\ref{sec:related}).

Despite the ubiquity of relational properties and recent logic-based or product-based approaches to reasoning with them (see Section~\ref{sec:related}), simple heap-manipulating examples like the following remain out of reach: 
\labf{C''}{
xp:=\NEW\ Int(1);\ \WHILE\; y>0\;\DO\; xp.set(xp.get() * y);\ y:=y-1\;\OD;\ x:=xp.get()
}
This Java-like program uses get/set procedures acting on an object that stores an integer value, and
$\splitbi{C}{C''}$ satisfies the same relational specification 
as (\ref{eq:factEq}).
This code poses significant new challenges.  
It is not amenable to product reductions that rely on renaming of identifiers to encode
two states as a single state: encoding of two heaps in one can be done, but at the cost of
significant complexity~\cite{Naumann06esorics} or exposing an underlying heap model below the level of abstraction of the programming language.  
Code like $C''$ also needs to be linked with implementations of the procedures it calls.
For reasoning about two versions of a module or library, 
relational hypotheses are needed, and calls need to be aligned to enable use of such hypotheses.

Floyd~\cite{Floyd67} articulates the fundamental method  of inductive assertions for partial correctness: establish that certain conditions hold at certain intermediate steps of computation, designating those conditions/steps by associating formulas with control flow points.
For relational reasoning, pairs of steps need to be aligned and it is again natural to designate those in terms of points in control flow.  
Alignment of steps has appeared 
in many guises in prior work, often implicit in simulation proofs 
but explicit in a few works~\cite{ZuckPGBFH05,BNR08secpriv,KovacsSF13}.

\emph{First contribution:}
In this paper we embody the alignment principle in a formal system at the level of abstraction of the programming language---as Hoare logic does for the inductive assertion 
method---with sufficient generality to encompass many uses of relational properties
for programs including procedures and dynamically allocated mutable objects.
Our logic (Section~\ref{sec:relLogic}) manifests the reasoning principle directly, in structured syntax.
It also embodies other reasoning principles, such as frame rules,
case analysis, 
and hypothetical specifications for procedures.
The rules encompass relations between both similarly- and differently-structured programs,
and handle partially and fully aligned iterations.
This achievement 
brings together ideas from many recent works (Section~\ref{sec:related}), 
together with two ingredients we highlight as contributions in their own right.

\emph{Second contribution:}
Our relational assertion language (Section~\ref{sec:relForm}) can describe agreement between unbounded pointer structures, allowing for differences in object allocation, as is needed to specify
noninterference~\cite{BanerjeeNaumann03b} and for simulation relations~\cite{BanerjeeNaumann02c} in languages like Java and ML where references are abstract.  Such agreements are
expressed without the need for recursively defined predicates, and the assertion language has a direct translation to SMT-friendly encodings of the heap.
(For lack of space we do not dwell on such encodings in this paper, 
which has a foundational focus,
but see~\cite{RosenbergBN12,BanerjeeNN15}.)

\emph{Third contribution:}
We introduce a novel form of ``biprogram'' (Section~\ref{sec:biprogramSem}) that makes explicit the reasoner's choice of alignments.
A biprogram run models an aligned pair of executions of the underlying programs.
The semantics of biprograms involves a number of subtleties:
To provide a foundation for extending the logic with encapsulation (based on ~\cite{RegLogJrnII}),
we need to use small-step semantics---which makes it difficult to prove
soundness of linking, even in the unary case~\cite{RegLogJrnII}.
For this to work we need to keep the semantics deterministic and to deal with 
semantics of hypotheses in judgments.

Section~\ref{sec:unaryLogic} provides background
and Section~\ref{sec:overview} is an overview of the logic using examples.
{\color{blue} 
This document is the technical report to accompany a paper to appear in FSTTCS 2016.
It has appendices and also some additional notes in the main body of the paper 
(which appear blue on color devices).
Sections~\ref{app:unary}--\ref{app:relProofRules} develop the syntax and semantics of the logic,
Sections~\ref{app:dissonant}--\ref{app:stack} develop examples,
Section~\ref{app:consistency} proves a theorem that says how biprogram runs model aligned pairs of ordinary runs,
and Section~\ref{app:sound} is on soundness of the logic.
There is a table of contents at the end of the document, 
}

\section{Background: synopsis of region logic}\label{sec:unaryLogic} 

For reasoning about the heap, separation logic is very effective,
with modal operators that implicitly describe heap regions.
But for relations on unbounded heap structures at the Java/ML level of abstraction we need explicit means to refer to heap regions, as in the dependency logic of Amtoft et al.~\cite{AmtoftBB06}.
Our relational logic is based on an underlying unary logic dubbed ``region logic'' (\dt{RL}),
developed in a series of papers~\cite{RegLogJrnI,RegLogJrnII,BanerjeeNN15} 
to which we refer for rationale and omitted details.
RL is a Hoare logic augmented with some side conditions (first order verification conditions) which facilitate local reasoning about frame conditions~\cite{RegLogJrnI}
in the manner of dynamic frames~\cite{kas:fac11,DBLP:conf/lpar/Leino10}.
In the logic such reasoning hinges on a frame rule. 
In a verifier, framing  
can be done by the VC-generator, optionally guided by annotation~\cite{RosenbergBN12}.
Stateful frame conditions also support an approach to encapsulation that validates a second order frame rule (at the cost of needing to use small-step semantics)~\cite{RegLogJrnII}.
Read effects enable the use of pure method calls in assertions and 
in frame conditions~\cite{BanerjeeNN15} and are useful for proving some
equivalences, like commuting assignments, that hold in virtue of disjointness of effects~\cite{Benton:popl04}.

The logic is formalized for imperative programs with first order procedures and dynamically allocated mutable objects (records), see Fig.~\ref{fig:bnf}. 
As in Java and ML, references are distinct from integers; they can be tested for equality but there is no pointer arithmetic.
Typing of programs is standard. 
In specifications we use ghost\fullstuff{\footnote{\fullstuff{We do not formalize a distinction between ghost and ordinary state.}}}
variables and fields of type $\Region$.
A \dt{region} is a set of object references, which may include the improper $\NULL$ reference.

\begin{figure*}[t]
  \begin{small}
\hfill
\(\begin{array}{l@{\hspace{.2em}}l@{\hspace{.3em}}r@{\hspace{.4em}}l}
  \multicolumn{4}{l}{m\in ProcName\hfill x,y,r\in VarName \hfill f,g\in FieldName \hfill K\in DeclaredClassNames
 } \\[1ex]
\mbox{(Types)} & T &\gassym&
\INT \gmid \BOOL \gmid \Region \gmid K 
\\[.2ex]
\mbox{(Program Expr.)}& E &\gassym& x \gmid c \gmid \NULL \gmid {E\oplus E } 
\quad \mbox{where $c$ is in $\Z$ and $\oplus$ is in $\{=,+,-,*,\geq,\land,\neg,\ldots\}$ }\\[.2ex]
\mbox{(Region Expr.)} & G  & \gassym & x \gmid \Emp \gmid \sing{E} \gmid G\Img f \gmid G \otimes G  
  \quad \mbox{where $\otimes$ is in $\{\cup,\cap,\setminus\}$ }\\[.2ex]
\mbox{(Expressions)} & F & \gassym & E \gmid G  \\[.2ex] 
\mbox{(Atomic comm.)} & A & \gassym & \skipc \gmid m() \gmid x := F \gmid x := \new{K} \gmid x := x.f \gmid x.f := x  \\[.2ex]
\mbox{(Commands)} & C &\gassym& A \gmid \letcom{m}{C}{C} \gmid \ifc{E}{C}{C} \gmid \whilec{E}{C} \gmid \seqc{C}{C} \\[.2ex] 

\mbox{(Biprograms)} & CC &\gassym& \splitbi{C}{C} \gmid \syncbi{A} 
  \gmid \letcombi{m}{\splitbi{C}{C}}{CC} 
\gmid \seqc{CC}{CC} \\[.2ex]
&&&
\gmid \ifcbi{E\smallSplitSym E}{CC}{CC} 
\gmid \whilecbiA{E\smallSplitSym E}{\P\smallSplitSym \P}{CC}
\end{array}\)
\end{small}
\vspace*{-1ex}
\caption{Programs and biprograms. 
Assume each class type $K$ has a declared list of fields, $\ol{f}:\ol{T}$.
Biprograms are explained in Section~\ref{sec:overview}.
}
\label{fig:bnf}
\end{figure*}

A \dt{specification} $\flowty{P}{Q}{\eff}$ is comprised of precondition $P$, postcondition $Q$, and frame condition $\eff$.  
Frame conditions include both read and write effects:
\[ \eff \gassym  \rd{x} \gmid \rd{G \Img f} \gmid \wri{x} \gmid \wri{G \Img f} 
        \gmid \eff,\eff \gmid  (empty) \]
The form $\rd{G\Img f}$ means the program may read locations $o.f$ where $o$ is a reference in the region denoted by expression $G$.
We write $\rw{x}$ to abbreviate the composite effect $\rd{x},\wri{x}$, and omit repeated tags:
$\rd{x,y}$ abbreviates $\rd{x},\rd{y}$.
Predicate formulas $P$ include standard first order logic with equality, 
region subset ($G\subseteq G$), and 
the  ``points-to'' relation $x.f=E$,
which says $x$ is non-null and the value of field $f$ equals $E$.
A \dt{correctness judgment}  has the form 
\(
\Phi\proves C: \flowty{P}{Q}{\eff}
\)
where the \dt{hypothesis context} $\Phi$ 
maps procedure names to specifications. 
In $C$ there may be \dt{environment calls} to procedures bound by $\keyw{let}$ inside $C$, and also \dt{context calls} to procedures in  $\Phi$.
%
The form $G\Img f$ is termed an \dt{image expression}.  
For an example of image expressions, consider this command which sums the elements of a singly-linked null-terminated list, ignoring nodes for which a deletion flag, $del$, has been set.
\labf{C_1}{
s:=0; 
\WHILE\; p\neq \NULL\; \DO\;
       \IF\; \neg p.del\; \THEN\; s:=s+p.val\;\FI;\;  p:=p.nxt\;\OD
}
For its specification we use ghost variable $r:\Region$ to contain the nodes. Its being closed 
under $nxt$ is expressed by $r\Img nxt \subseteq r$ in this specification: 
\begin{equation}
\label{eq:sumspec}
p \in r \land r\Img nxt \subseteq r 
\; \leadsto \;
s = sum(listnd(\keyw{old}(p)))\;[ 
  \rw{s,\,p},\; \rd{r,\, r\Img val,\, r\Img nxt,\, r\Img del} ] 
\end{equation}
The r-value of the image expression $r\Img nxt$ is the set of values of $nxt$ fields
of the objects in $r$.
In frame conditions, expressions are used for their l-values.
In this case, the frame condition uses image expressions to say that for any object $o$ in $r$, locations $o.val, o.nxt, o.del$ may be read.
The frame condition also says that variables $s$ and $p$ may be both read and written.
Let function $listnd$ give the mathematical list of non-deleted values.\fullstuff{\footnote{\fullstuff{We do not formalize $\keyw{old}$ expressions in the logic, but our uses of them can be desugared using ghost variables.}}}

Some proof rules in RL have side conditions which are first order formulas on one or two states. 
\fullstuff{In this paper we treat these subsidiary judgments semantically. (Cognoscenti will have no difficulty thinking of ways to encode the subsidiary judgments as $\forall$-formulas amenable to SMT, for usual representations of program state.)
}
One kind of side condition, dubbed the ``frames judgment'', delimits the part of state on which a formula depends (its read effect). 
RL's use of stateful frame conditions provides for a useful frame rule, and even second order frame rule~\cite{OHearnYangReynoldsInfoToplas,RegLogJrnII}, but there is a price to be paid.  Frame conditions involving state dependent region expressions are themselves susceptible to interference by commands.  That necessitates side conditions, termed ``immunity'' and ``read-framed'', in the proof rules for sequence and iteration \cite{RegLogJrnII,BanerjeeNN15}.
The frame rule allows to infer from 
$\Phi\HPflowtr{}{}{P}{C}{Q}{\eff}$
the conclusion
$\Phi\HPflowtr{}{}{P\land R}{C}{Q\land R}{\eff}$
provided that $R$ is framed by read effects $\effe$ (written $\fra{\effe}{R$})
for locations disjoint from those writable according to $\eff$ (written $\ind{\effe}{\eff}$).

In keeping with our goal to develop a comprehensive deductive system, our unary and relational logics include a rule for discharging hypotheses, expressed in terms of the linking construct.
Here is the special case of a single non-recursive procedure.
\[
\inferrule*[left=Link]{ 
\sflowtr{R}{m}{S}{\effe}\HPflowtr{}{}{P}{ C }{Q}{\eff} \\
\HPflowtr{}{}{R}{B}{S}{\effe}
 }{
\HPflowtr{}{}{P}{ \letcom{m}{B}{C} }{Q}{\eff}
}
\]

\section{Overview of the relational logic}\label{sec:overview}

This section sketches highlights of relational reasoning about a number of illustrative examples,
introducing features of the logic incrementally.
Some details are glossed over.

We write $\splitbi{C}{C'} : \Q\rspecSym \R$ to express that a pair of programs $C,C'$ satisfies the relational contract with precondition $\Q$ and postcondition $\R$, 
leaving aside frame conditions for now.
The judgment constrains executions of $C$ and $C'$ from pairs of states related by $\Q$.
(For the grammar of relational formulas, see (\ref{eq:relFormulas}) in Section~\ref{sec:relForm}.)
It says neither execution faults (e.g., due to null dereference), and if both terminate then the final states are related by $\R$.  Moreover no context procedure is called outside its precondition.
(We call this property the $\forall\forall$ form, for contrast with refinement properties of $\forall\exists$ form.)

Assume $f,g$ are pure functions.
The programs 
\[
C_0 \:\eqdef\:  x := f(z); y := g(z) \qquad C_0' \:\eqdef\: y := g(z); x := f(z)
\]
are equivalent.
Focusing on relevant variables, the equivalence can be specified as
\begin{equation}\label{eq:reorder}
\Splitbi{C_0}{C_0'} : \: z\eqbi z \rspecSym x\eqbi x \land y\eqbi y 
\end{equation}
which  can be proved as follows.
Both $C_0$ and $C_0'$ satisfy
\( true \leadsto x=f(z) \land y=g(z) \),
which directly entails that 
\( 
\Splitbi{C_0}{C_0'} : \: \Both{true} \rspecSym  \Both{(x=f(z) \land y=g(z))}
\) 
by an embedding rule.  
The general form of embedding combines two different unary judgments, with different specifications, using relational formulas that assert a predicate on just the left 
($\leftF{}$) or right ($\rightF{}$) state.
So  $\Both{P}$ is short for $\leftF{P}\land \rightF{P}$.
Since $z$ is not written by $C_0$ or $C_1$, we can introduce $z\eqbi z$ using the relational frame rule,
to obtain
\( 
\Splitbi{C_0}{C_0'}  : \: z\eqbi z \rspecSym \Both{(x=f(z) \land y=g(z))} \land z\eqbi z
\).
This yields (\ref{eq:reorder}) using the relational rule of consequence 
with the two valid relational assertion schemas
$u\eqbi u' \land \leftF{(u=v)} \land \rightF{(u'=v')} \imp v\eqbi v'$
and 
$z\eqbi z \imp f(z) \eqbi f(z)$.

For the factorial example $\splitbi{C}{C'}$ in Section~\ref{sec:intro},  
we would like to align the loops and use the simple
relational invariant $x\eqbi x \land y \eqbi y+1$.
We consider the form $\splitbi{C}{C'}$ as a biprogram 
which can be rewritten to equivalent forms using the 
\dt{weaving} relation which preserves the underlying programs 
but aligns control points together so that relational assertions can be used. (A minor difference from most other forms of product program is that we do not need to rename apart the variables on the left and right.)
The weaving relation is given in Section~\ref{sec:biprogramSem}.
In this case we weave to the form
\[\begin{array}{l}
\splitbi{x:=1\,}{\,x:=1; y:=y-1};\,
\WHILE\;y>0 \mid y \geq 0\; \DO \;
\Splitbi{x:= x * y }{ x:= x * y + x}; \, \syncbi{y:= y-1} 
\end{array}
\]
This enables us to assert the relational invariant at the beginning and end of the loop bodies.
Indeed, we can also assert it just before the last assignments to $y$.
The rule for this form of loop requires the invariant to imply equivalence of the two loops' guard conditions, 
which it does:
$x\eqbi x \land y \eqbi y+1 \imp (y>0  \eqbi y \geq 0)$.
For a biprogram of the \dt{split} form $\splitbi{C}{C'}$,
the primary reasoning principle is the lifting of unary judgments about $C$ and $C'$.
For an atomic command $A$, the \dt{sync} notation $\syncbi{A}$ is an alternative to $\splitbi{A}{A}$ 
that indicates its left and right transition are considered together.
This enables the use of relational specifications for procedures,
and a relational principle for object allocation.
For an ordinary assignment, sync merely serves to abbreviate, as in 
$\syncbi{y:= y-1}$ above.

The next example involves the heap and it also involves a loop that is ``dissonant'' in the sense that we do not want to align all iterations---that is, \emph{alignment is ultimately about traces,} not program texts.
Imagine the command $C_1$ from Section~\ref{sec:unaryLogic} is run on a list from which secret values have been deleted. 
To specify that no secrets are leaked, we use the relational judgment
\( \splitbi{C_1}{C_1}: \: listnd(p) \eqbi listnd(p) \rspecSym s \eqbi s \)
which says: Starting from any two states containing the same non-deleted values, terminating computations 
agree on the sums.  
The judgment can be proved by showing the functional property that $s$ ends up as 
$sum(listnd(\keyw{old}(p)))$.
But we can avoid reasoning about list sums and prove this relational property by aligning some of
the loop iterations in such a way that $listnd(p)\eqbi listnd(p) \land s\eqbi s$ holds at every aligned pair, that is, it is a relational invariant.   
Not every pair of loop iterations should be aligned: 
When $p.del$ holds for the left state but not the right,
a left-only iteration maintains the invariant, and \emph{mutatis mutandis} when $p.del$ holds only on the right.
To handle such non-aligned iterations we use a novel syntactic annotation dubbed \dt{alignment guards}. The idea is that the loop conditions are in agreement, and thus the iterations are synchronized, unless one of the alignment guards hold---and then that iteration is  unsynchronized but the relational invariant must still be preserved.
We weave $\splitbi{C_1}{C_1}$ to the form
\begin{equation}\label{eq:C1woven}
\syncbi{s:=0};
\!\!
\begin{array}[t]{l} 
\WHILE\ p\neq\NULL\mid p\neq\NULL\ \alignSepSym\ \leftF{(p.del)} \mid \rightF{(p.del)}\\
\quad \DO \; 
\IF\; \neg p.del \mid \neg p.del \; 
   \THEN\; \syncbi{s:=s+p.val} \;\FI;\; \syncbi{p:=p.nxt}\;\OD
\end{array}
\end{equation}
with alignment guards $\leftF{p.del}$ and $\rightF{p.del}$.
The rule for the while biprogram 
has three premises for the loop body: for executions on the left (resp.\ right) under alignment guard $\leftF{p.del}$ (resp.\ $\rightF{p.del}$) and for simultaneous executions when neither of the alignment guards hold. 
Each premise requires the invariant to be preserved.
\fullstuff{The loop body uses the synchronized conditional which requires agreement on the guard conditions; indeed, agreement does hold when neither of the loop alignment guards hold.}

The final example is a change of data representation.  It illustrates dynamic allocation and frame conditions, as well as procedures and linking.  
A substantive example of this sort would be quite lengthy, so we contrive a toy example to provide hints of the issues that motivate 
various elements of our formal development.
Our goal is to prove a conditional equivalence between these programs,
whose components are defined in due course.
\[ 
C_4  \eqdef \letcomMulti{push(x:\INT)=B}{Cli} \qquad
C_4' \eqdef \letcomMulti{push(x:\INT)=B'}{Cli}
\]
These differ only in the implementations $B,B'$ of the stack interface (here stripped down to a single procedure), to which the client program $Cli$ is linked.
For modular reasoning, the unary contract for $push$ should not expose details of the data representation.
We also want to avoid reliance on strong functional specifications---the goal is equivalence of the two versions, not functional correctness of the client.
The client, however, should respect encapsulation of the stack representation,
to which end frame conditions are crucial.
A simple pattern is for contracts to expose a ghost variable $rep$ (of type $\Region$) for the set of objects considered to be owned by 
a program module.
Here is the specification for $push$, with parts named for later reference. Let $size$ and $rep$ be \dt{spec-public}, i.e., they can be used in public contracts but not in client code~\cite{LeavensM07}.
\begin{equation}
\label{eq:push}
\begin{array}[t]{l@{\,}l}
push(x:\INT): \, R \leadsto S [\effe] \mbox{ where}
& R \eqdef size < 100 \\
& S \eqdef size = \keyw{old}(size) + 1 \\
& \effe \eqdef \rw{rep,size,rep\Img\allfields} 
\end{array}
\end{equation}
Variables $rep$ and $size$ can be read and written 
(keyword $\keyw{rw}$) by $push$. 
This needs to be explicit, 
even though client code cannot access them, because reasoning about client code involves them.
The notation $rep\Img\allfields$ designates all fields of objects in $rep$; these too may be read and written.
The specification makes clear that calls to $push$ affect the encapsulated state, while not exposing details.
Here is one implementation of $push(x)$.
\labf{B}{
\!\!\!
\begin{array}[t]{l}
top := \NEW\ Node(top,x);\ rep := rep \union \{top\};\ size\increment
\end{array}
}
Variable $top$ is considered internal to the stack module, so it need not appear 
in the frame condition.
The alternate implementation of $push$ replaces $top$ by 
module variables $free: \INT;\ slots: String[\: ];$.
\labf{B'}{
\!\!\!
\begin{array}[t]{l}
\IF\ slots = \NULL\ \THEN\ slots := \NEW\ String[100];\ rep := rep \union \{slots\};\ free := 0\ \FI; \\
slots[free\increment] := x;\  size\increment
\end{array}
}
Unary correctness of the two versions is proved using module invariants
 \labf{I}{
\!\!\!
\begin{array}[t]{l}
(top=\NULL \land size=0) \; \lor 
(top\in rep \land rep\Img nxt \subseteq rep \land size=\mathconst{length}(\mathconst{list}(top))) 
   \end{array}
}

\vspace*{-2ex}

 \labf{I'}{(slots=\NULL\land  size=0) \lor (slots\in rep \land size = free) }
Here $\mathconst{list}(top)$ is the mathematical list of values reached from $top$.
Recall that in an assertion the expression $rep\Img nxt$ is the image of set $rep$ under the $nxt$ field, i.e., the set of values of $nxt$ fields of objects in $rep$.  
The condition $rep\Img nxt \subseteq rep$ says that $rep$ is closed under $nxt$.
This form is convenient in using ghost code to express shapes of data structures without recourse to reachability or other inductive predicates~\cite{RegLogJrnI,RosenbergBN12}.


As a specific $Cli$,  we consider one that
allocates and updates a node of the same type as used by the list implementation; this gets assigned to a global variable $p$.  
\labf{Cli}{ 
   push(1);\ p := \NEW\ Node(\NULL,2);\ 
   p.val := 3;\ push(4) 
}
Having completed the definitions of $C_4,C_4'$ we can ask:
In what sense are $C_4,C_4'$ 
equivalent?
A possible specification for $\splitbi{C_4}{C_4'}$ 
requires agreement on $size$ and ensures agreement on $size$ and on $p$ and $p.val$.
However, the latter agreements cannot be literal equality: following 
the call $push(1)$, one implementation has allocated a $Node$ whereas the array implementation has not.  Depending on the allocator, different references may be assigned to $p$ in the two executions.
The appropriate relation is ``equivalence modulo renaming of references''~\cite{AmtoftBB06,BanerjeeNaumann02c,BanerjeeNaumann03b,BentonHN14,BentonKBH07}.
For region expression $G$ and field name $f$,
we write $\Agr G\Img f$ for the \dt{agreement} relation that says 
there is a partial bijection on references between the two states,
that is total on the region $G$, and for which corresponding $f$-fields are equal.
The notation $\Agr G\Img\allfields$ means agreement on all fields.
In the present example, the only region expression used is the singleton region $\sing{p}$ containing the reference denoted by $p$.  


To prove a relational judgment 
for $\splitbi{C_4}{C_4'}$
we need suitable relational judgments for 
$\splitbi{B}{B'}$ for the implementations of $push$.
It is standard~\cite{Hoare:data:rep} that they should 
preserve a ``coupling relation''
that connects the two data representations and also includes 
the data invariants for each representation.
For the example, the connection is that 
the sequence of elements reached from $top$, written $\mathconst{list}(top)$, is the same as the reversed sequence of elements in $slots[0..free-1]$. Writing $rev$ for reversal,
we define the coupling and specification 
\[ 
\L\eqdef \leftF{I} \land \rightF{I'} \land LtR 
\qquad
LtR \eqdef 
list(top) \eqbi 
            rev(\langle\,\rangle \,\keyw{if}\, slots = \NULL \,\keyw{else}\, slots[0..free-1])
\]
\vspace*{-5ex}
\begin{equation}\label{eq:ggx}
\splitbi{C_4}{C_4'} : \;
\Both{(size=0)} \land \L
\:\rspecSym \:
p\eqbi p \land size\eqbi size
\land \Agr \sing{p}\Img \allfields \land \L
\end{equation}
%
%
We now proceed to sketch a proof of (\ref{eq:ggx}).
First, we weave $\splitbi{C_4}{C_4'}$ to 
\(  \letcombi{push(x:\INT)}{\splitbi{B}{B'}}{ \Syncbi{Cli} }
\).
Here $\Syncbi{Cli}$ abbreviates the fully aligned biprogram
\( \syncbi{push(1)}; 
  \syncbi{p := \NEW\ Node(\NULL,2)};
  \syncbi{p.val := 3}; 
  \syncbi{push(4)} 
\). 
This biprogram 
simultaneously links the procedure bodies on left and right, and 
aligns the client.
Using $\syncbi{p := \NEW\ Node(\NULL,2)}$ enables use of a
relational postcondition that says the objects are in agreement.
Using $\syncbi{push(4)}$ enables use of $push$'s relational specification.

Like in unary RL, the proof rule for linking
has two premises:
one says the bodies $\splitbi{B}{B'}$ satisfy their specification, 
the other says $\Syncbi{Cli}$ satisfies the overall specification 
under the hypothesis that $push$ satisfies its spec
(see \rn{rLink} in Fig.~\ref{fig:proofrulesR}). 
This hypothesis context gives $push$ a relational specification,
using $\Agr x$ as sugar for $x\eqbi x$:
\label{disp:Phi:pg}
\labf{\Phi}{
push(x): \;
\begin{array}[t]{l}
   \Both{R}
   \land \Agr size \land \Agr x
   \land \L
\:\rspecSym\:
   \Both{S} 
   \land \Agr size
   \land \L
\; [\effe, \rw{top} 
\mid 
\effe, \rw{slots,free}] 
\end{array}
}
Here $\effe$ is the effect $\rw{rep,size,rep\Img\allfields}$
in the original specification~(\ref{eq:push}) of $push$.

The specification in $\Phi$ is not simply a relational lift of $push$'s public specification~(\ref{eq:push}). 
Invariants $I$ and $I'$ on internal data structures should not appear in $push$'s API:
they should be hidden, because the client should not touch the 
internal state on which they depend.  Effects on module variables (like $top$) should also be hidden.
This kind of reasoning is the gist of second order framing~\cite{OHearnYangReynoldsInfoToplas,RegLogJrnII}.
The relational counterpart is a relational second order frame rule which says that any client that respects encapsulation will preserve $\L$.
Hiding is the topic of another paper, for which this one is laying the groundwork (see Section~\ref{sec:discuss}).

\section{Relational formulas}\label{sec:relForm} 

The relational assertion language is essentially syntax for a first order structure comprised of the variables and heaps of two states, together with a refperm connecting the states.
\begin{equation}\label{eq:relFormulas}
\begin{array}{ll}
 \P \, \gassym & 
    F\eqbi F \gmid \Agr G\Img f \gmid \later\P \gmid \leftF{P} \gmid \rightF{P} \gmid
 \P\land\P \gmid \P\imp\P 
 \gmid \all{x\smallSplitSym x':K}{\P}
 \end{array}
\end{equation}
A \dt{refperm} is a type-respecting partial bijection from references allocated in one state to references allocated in the other state.
For use with SMT provers, a refperm can be encoded by a pair of maps with universal formulas stating they are inverse~\cite{BanerjeeNN15}. 
The syntax for relations caters for dynamic allocation
by providing primitives such as $F\eqbi F'$ that says the value of $F$ in the left state
equals that of $F'$ in the right state, modulo the refperm.
In case of integer expressions, this is ordinary equality.  
For reference expressions, 
it means the two values are related by the refperm.  
For region expressions, $G\eqbi G'$ means the refperm forms a bijection between the 
reference set denoted by $G$ in the left state and $G'$ in the right state (ignoring $\semNull$).
The agreement  formula $\Agr G\Img f$ says, of a pair of states, that the refperm is total on 
the set denoted by $G$ in the left state,
and moreover the $f$-field of each object in that set has the same value, modulo refperm,
as the $f$-field of its corresponding object in the right state.

For commands that allocate, the postcondition needs to allow the refperm to be extended,
which is expressed by the modal operator $\later$ (read ``later''):
$\later\P$ holds if there is an extension of the refperm with zero or more pairs of references for which $\P$ holds.
For example, after the assignment to $p$ in the stack example, 
the relational rule for allocation yields postcondition $\later(p\eqbi p\land \Agr \{p\}\Img\allfields)$.
Aside from the left and right embeddings of unary predicates ($\leftF{P}$ and $\rightF{P}$),
the only other constructs are the logical ones (conjunction, implication, quantification over values).

\fullstuff{
We use the following syntax sugars:
$\Both{P} \eqdef \leftF{P} \land \rightF{P}$,
$\Agr E \eqdef E\eqbi E$,
$\Agr (\rd{G\Img f}, \rd{x}) \eqdef \Agr G\Img f \land \Agr x$, etc.
Note that $\Agr E$ is unambiguous, but we cannot use the 
analogous abbreviation for region expressions:
For region expressions  of the image form, $G\Img f$, the atomic formula 
$\Agr G\Img f$ already has a meaning, which is different from
$G\Img f \eqbi G\Img f$.
The meaning of 
$G\Img f \eqbi G\Img f$ is equality, modulo refperm, of two sets:
the $f$-fields of $G$-objects in the left state and in the right state.
By contrast, $\Agr G\Img f$ means that for each non-null reference $o$ in region $G$ (interpreted in the left state), with counterpart $o'$ according to the refperm, the $f$ fields of $o$ and $o'$ agree.
}

Let $\always\P \eqdef \neg\later\neg\P$.
Validity of $\P\imp\always\P$ is equivalent to $\P$ being \dt{monotonic},
i.e., not falsified by extension of the refperm.
\fullstuff{Validity of $\later\P\imp\P$ expresses that $\P$ is refperm-independent.}
Here are some valid schemas:
$\P\imp\later\P$, 
$\later\later\P\imp\later\P$,
and
$\later(\P\land\Q) \imp \later\P\land\later\Q$.
The converse of the latter is not valid.\fullstuff{\footnote{\fullstuff{For example, 
$\later (x\eqbi y) \land \later(x\eqbi z\land \rightF{(z\neq y)})$ is satisfiable
but $\later (x\eqbi y \land x\eqbi z\land \rightF{(z\neq y)})$ is not.}}}
For framing, a key property is that
$\later\P\land\Q \imp \later(\P\land\Q)$ is valid if $\Q$ is monotonic.
In practice, $\later$ is only needed in postconditions, and only at the top level.
Owing to $\later\later\P\imp\later\P$, this works 
fine with sequenced commands.\fullstuff{\footnote{\label{fn:later}\fullstuff{There is a convenient 
derived rule for sequencing of judgments like 
$\rflowtrf{\P}{CC}{\later\Q}$ and
$\rflowtrf{\Q}{DD}{\later\R}$.
We can use rule \rn{rLater} to get 
$\rflowtrf{\later\Q}{DD}{\later\later\R}$, and thus 
$\rflowtrf{\later\Q}{DD}{\later\R}$ by the rule of consequence, 
using $\later\later\R\imp\later\R$.
Then by the sequence rule 
we get $\rflowtrf{\P}{CC;DD}{\later\R}$.
}}}  
Many useful formulas are monotonic, including 
$\Agr G\Img f$ and $F\eqbi F'$,
but not $\neg (F\eqbi F')$.
\fullstuff{The $\later$ operator can also break monotonicity: $\later (x\eqbi x)$ is not 
monotonic.
The $\later$ operator extends the refperm but not the sets of allocated references.  
So this is valid: $\lloc\eqbi\lloc \land \later\P \imp \P$, where $\land$ binds more tightly than $\imp$.
(Because $\lloc\eqbi\lloc$ says the refperm is a total bijection on allocated references and has no proper extensions.)
}

\section{Biprograms}\label{sec:biprogramSem} 

A biprogram $CC$ (Fig.~\ref{fig:bnf}) represents a pair of commands, which are given by syntactic projections
defined by clauses including the following:
$\Left{\splitbi{C}{C'}}  \eqdef  C$,
$\Right{\splitbi{C}{C'}}  \eqdef  C'$,
$\Left{\syncbi{A}}  \eqdef A$,
$\Left{\ifcbi{E\smallSplitSym E'}{BB}{CC}}  \eqdef  \ifc{E}{\Left{BB}}{\Left{CC}}$,
and 
$\Left{\letcombi{m}{\splitbi{C}{C'}}{CC}}  \eqdef  \letcom{m}{C}{\Left{CC}}$.
%
%
The weaving relation has clauses including the following.
\[\begin{array}{l}
\splitbi{A}{A} \weave \syncbi{A} \qquad \mbox{(for atomic commands $A$)}\\[.5ex]
\Splitbi{C;D}{C';D'} \weave \splitbi{C}{C'};\splitbi{D}{D'} \\[.5ex]
\Splitbi{ \ifc{E}{C}{D} }{ \ifc{E'}{C'}{D'} } 
   \weave \ifcbi{E\smallSplitSym E'}{\splitbi{C}{C'} }{ \splitbi{D}{D'} }
\\[.5ex]
\Splitbi{ \whilec{E}{C} }{ \whilec{E'}{C'} }
   \weave \whilecbiA{E\smallSplitSym E'}{\P\smallSplitSym \P'}{\splitbi{C}{C'}}
\qquad \mbox{(for any $\P,\P'$)}
\\[.5ex]
 \end{array}\]
Additional clauses are needed for congruence, e.g., 
$CC\weave DD$ implies $BB;CC\weave BB;DD$.
The loop weaving introduces chosen alignment guards.
The \dt{full alignment} of a command $C$ is written $\Syncbi{C}$ and defined by
$\Syncbi{A} \eqdef \syncbi{A}$, 
$\Syncbi{C;D} \eqdef \Syncbi{C};\Syncbi{D}$,
$\Syncbi{\ifc{E}{C}{D}} \eqdef \ifcbi{E\smallSplitSym E}{\Syncbi{C}}{\Syncbi{D}}$,
$\Syncbi{\whilec{E}{C}} \eqdef \whilecbiA{E\smallSplitSym E}{\False\smallSplitSym 
 \False}{\Syncbi{C}}$, etc.
Note that $\splitbi{C}{C}\weave^* \Syncbi{C}$ for any $C$.

Commands are deterministic (modulo allocation), so 
termination-insensitive noninterference and equivalence properties can be expressed in a simple $\forall\forall$ form described at the start of Section~\ref{sec:overview}, rather than the $\forall\exists$ form needed for
refinement and for possibilistic noninterference
(``for all runs \ldots there exists a run \ldots'').
The transition rules for biprograms must ensure that the behavior is compatible with the underlying unary semantics, while enforcing the intended alignment.
That would still allow some degree of nondeterminacy in biprogram transitions.
However, we make biprograms deterministic (modulo allocation), because it greatly simplifies the soundness proofs.
Rather than determinize by means of a scheduling oracle or other artifacts
that would clutter the semantics,
we build determinacy into the transition semantics. 
Whereas the syntax aligns points of interest in control flow, 
biprogram traces explicitly represent aligned pairs of executions. 
We make the arbitrary choice of left-then-right semantics for the split form. 
In a trace of $\splitbi{C}{C'}$, every step taken by $C$ is effectively aligned with the initial state for $C'$.  This is followed by the steps of $C'$, each aligned with the final state of $C$.
To illustrate the idea, here is a sketch of the trace of a split biprogram (center column) and
its alignment with left and right unary traces.

\vspace*{-2ex}

\begin{center}
  \begin{small}
  \begin{tikzpicture}[
    a1/.style={thin, dash pattern=on 2pt off 2pt}
    ]
    \matrix[column sep=5mm]{
      \node (U0) {$\configc{\mbox{x:=0; y:=0}}$}; 
      & \node (T0) {$\configc{\mbox{(x:=0; y:=0 $|$ x:=0; y:=0)}}$}; 
      & \node (V0) {$\configc{\mbox{x:=0; y:=0}}$}; \\[-1ex]
     \node (U1) {$\configc{\mbox{y:=0}}$}; 
      & \node (T1) {$\configc{\mbox{(y:=0 $|$ x:=0; y:=0)}}$}; \\[-1ex]
      \node (U2) {$\configc{\mbox{skip}}$}; 
      & \node (T2) {$\configc{\mbox{(skip $|$ x:=0; y:=0)}}$}; \\[-1ex]
      & \node (T3) {$\configc{\mbox{(skip $|$  y:=0)}}$}; 
      & \node (V1) {$\configc{\mbox{y:=0}}$}; \\[-1ex]
      & \node (T4) {$\configc{\syncbi{\mbox{skip}}}$}; 
      & \node (V2) {$\configc{\mbox{skip}}$}; \\[-1ex]
    };
    \draw[-] (U0) edge[a1] (T0);
    \draw[-] (T0) edge[a1] (V0);
    \draw[-] (U1) edge[a1] (T1);
    \draw[-] (T1) edge[a1, bend right=5] (V0);
    \draw[-] (U2) edge[a1] (T2);
    \draw[-] (T2) edge[a1, bend right=10] (V0);
    \draw[-] (U2) edge[a1, bend right=5] (T3);
    \draw[-] (T3) edge[a1] (V1);
    \draw[-] (U2) edge[a1, bend right=10] (T4);
    \draw[-] (T4) edge[a1] (V2);
  \end{tikzpicture}
  \end{small}
\end{center}

\vspace*{-2ex}

\noindent
This pattern is also typical for ``high conditionals'' in noninterference proofs,
where different branches may be taken (cf.\ rule \rn{rIf4}).
Here is the sync'd version in action.

\vspace*{-2ex}

\begin{center}
  \begin{small}
  \begin{tikzpicture}[
     a1/.style={thin, dash pattern=on 2pt off 2pt}
    ]
    \matrix[column sep=5mm]{
      \node (U0) {$\configc{\mbox{x:=0; y:=0}}$}; 
      & \node (T0) {$\configc{\syncbi{\mbox{x:=0}}; \syncbi{\mbox{y:=0}}}$}; 
      & \node (V0) {$\configc{\mbox{x:=0; y:=0}}$}; \\[-1ex]
      \node (U1) {$\configc{\mbox{y:=0}}$}; 
      & \node (T1) {$\configc{\syncbi{\mbox{y:=0}}}$}; 
      & \node (V1) {$\configc{\mbox{y:=0}}$}; \\[-1ex]
      \node (U2) {$\configc{\mbox{skip}}$}; 
      & \node (T2) {$\configc{\syncbi{\mbox{skip}}}$}; 
      & \node (V2) {$\configc{\mbox{skip}}$}; \\[-1ex]
    };
    \draw[-] (U0) edge[a1] (T0);
    \draw[-] (T0) edge[a1] (V0);
    \draw[-] (U1) edge[a1] (T1);
    \draw[-] (T1) edge[a1] (V1);
    \draw[-] (U2) edge[a1] (T2);
    \draw[-] (T2) edge[a1] (V2);
  \end{tikzpicture}
  \end{small}
\end{center}

The relational correctness judgment has the form 
$\Phi \rHPflowtr{}{}{\P}{CC}{\Q}{\eff|\eff'} $.
The hypothesis context $\Phi$ maps some procedure names to their specifications:
$\Phi(m)$ may be a unary specification as before or else a relational one of the form 
$\rflowty{\R}{\S}{\eff|\eff'}$.
Frame conditions retain their meaning, separately for the left and the right side.
In case $\eff$ is the same as $\eff'$, the judgment or specification is abbreviated as 
$\rflowty{\P}{\Q}{\eff}$.

The semantics of biprograms uses small steps, which makes alignments explicit.
A configuration is comprised of a biprogram, two states, and two environments for procedures.
The transition relation depends on a semantic interpretation for each procedure in the hypothesis context $\Phi$.  Context calls, i.e., calls to procedures in the context, take a single step in accord with the interpretation.
For the sake of determinacy, this is formalized in the semantics of relational correctness 
by quantifying over deterministic ``interpretations'' of the specifications (as in \cite{BanerjeeNN15}),
rather than a single nondeterministic transition rule (as in  \cite{RegLogJrnII,OHearnYangReynoldsInfoToplas}).

\fullstuff{An aligned conditional, $\ifcbi{E\smallSplitSym E'}{CC}{DD}$,
faults from initial states that do not agree on the guard conditions $E,E'$.
An aligned loop 
$\whilecbiA{E\smallSplitSym E'}{\P\smallSplitSym \P'}{CC}$
executes the left part of the body, $\Left{CC}$, if $E$ and the left alignment guard $\P$ both hold,
and \emph{mutatis mutandis} for the right.  If neither alignment guard holds,
the loop faults unless the guards $E,E'$ agree.

The relational correctness judgment disallows faults, so
correctness of a biprogram implies it represents the intended alignments.
Note that the weaving transformations can introduce, but not eliminate, alignment faults.
}

Let us sketch the semantic consistency theorem,
which confirms that executions of a biprogram from a pair of states
correspond to pairs of executions of the underlying commands,
so that judgments about biprograms represent relational properties
of the underlying commands.
Suppose  $\Phi\rHPflowtr{}{}{\P}{\splitbi{C}{C'}}{\Q}{\eff|\eff'}$ is valid and 
$\Phi$ has only unary specifications.
Consider any states $\sigma,\sigma'$ that are related by $\P$ (modulo some refperm).
Suppose $C$ and $C'$, when executed from $\sigma,\sigma'$, reach final states $\tau,\tau'$.
(In the formal semantics, transitions are defined in terms of interpretations $\phi$ that satisfy the specifications $\Phi$, so this is written
$\config{C}{\sigma}\tranStar{\phi}\config{\skipc}{\tau}$ and 
$\config{C'}{\sigma'}\tranStar{\phi}\config{\skipc}{\tau'}$.)
Then $\tau,\tau'$ satisfy $\Q$.

\section{Relational region logic}\label{sec:relLogic} 

\begin{figure}[t!]
\begin{small}
\begin{mathpar}
\inferrule*[left=rLink]{ 
\rflowtr{\R}{m}{\S}{\effe} 
   \rHPflowtr{}{}{\P}{\Syncbi{C}}{\Q }{\eff} \\
\rHPflowtr{}{}{\R}{\splitbi{B}{B'}}{\S}{\effe} 
}{ 
 \rHPflowtr{}{}{\P}{
  \letcombi{m}{\splitbi{B}{B'}}{\Syncbi{C}}}{\Q}{\eff}
}

\and
\inferrule*[left=rIf4]
{
\Phi \rHPflowtr{}{}{\P\land\leftF{E}\land\rightF{E'}}{\splitbi{C}{C'}}{\Q}{\eff|\eff'} \\
\Phi \rHPflowtr{}{}{\P\land\leftF{E}\land\rightF{\neg E'}}{\splitbi{C}{D'}}{\Q}{\eff|\eff'} \\
\Phi \rHPflowtr{}{}{\P\land\leftF{\neg E}\land\rightF{E'}}{\splitbi{D}{C'}}{\Q}{\eff|\eff'} \\
\Phi \rHPflowtr{}{}{\P\land\leftF{\neg E}\land\rightF{\neg E'}}{\splitbi{D}{D'}}{\Q}{\eff|\eff'} 
}{
\Phi\rHPflowtr{}{}{\P}{\splitbi{\ifc{E}{C}{D}}{\ifc{E'}{C'}{D'}}}{\Q}{\eff,\ftpt(E)|\eff',\ftpt(E')}
}

\and
\inferrule*[left=rIf]
{
\P\imp E\eqbi E' \\
\Phi \rHPflowtr{}{}{\P\land\leftF{E}\land\rightF{E'}}{CC}{\Q}{\eff|\eff'} \\
\Phi \rHPflowtr{}{}{\P\land\leftF{\neg E}\land\rightF{\neg E'}}{DD}{\Q}{\eff|\eff'} 
}{
\Phi \rHPflowtr{}{}{\P}{\ifcbi{E\smallSplitSym E'}{CC}{DD}}{\Q}{\eff,\ftpt(E)|\eff',\ftpt(E')}
}

\and
\inferrule*[left=rWeave]{
  \Phi \rHPflowtr{}{}{\P}{DD}{\Q}{\eff|\eff'} \\
  CC \weave DD \\
  unaryOnly(\Phi) \\
  terminates(\Left{\P},\Left{DD}) \\
  terminates(\Right{\P},\Right{DD})
}{ 
  \Phi \rHPflowtr{}{}{\P}{CC}{\Q}{\eff|\eff'} \\
}

\end{mathpar}
\end{small}
\vspace*{-3ex}
\caption{Selected relational proof rules.
}
\label{fig:proofrulesR}
\end{figure}

Selected proof rules appear in Fig.~\ref{fig:proofrulesR}.

For linking a procedure with its implementation, rule \rn{rLink} caters for a client program $C$ related to itself, in such a way that its executions can be aligned to use the same pattern of calls.
The procedure implementations may differ, as in the stack example, Section~\ref{sec:overview}.
The rule shown here is for the special case of a single procedure, 
and the judgment for $\splitbi{B}{B'}$ has empty hypothesis context, to disallow recursion.
We see no difficulty to add mutually recursive procedures, as done for the unary logic in~\cite{RegLogJrnII}, but have not yet included that in a detailed soundness proof.
The soundness proof is basically an induction on steps
as in~\cite{RegLogJrnII} but with the construction 
of an interpretation as in the proof of the linking rule in \cite{BanerjeeNN15}.
The general rule also provides for un-discharged hypotheses for ambient libraries used in the client and in the procedure implementations~\cite{RegLogJrnII}.

Rule \rn{rIf4} is the obvious rule that considers all paths for a conditional not aligned with itself (e.g., for ``high branches''), whereas \rn{rIf} leverages the alignment 
designated by the biprogram form.
The disjunction rule---i.e., from 
  $\Phi \rHPflowtr{}{}{\P_0}{CC}{\Q}{\eff|\eff'}$ and
  $\Phi \rHPflowtr{}{}{\P_1}{CC}{\Q}{\eff|\eff'}$ 
  infer $\Phi \rHPflowtr{}{}{\P_0\lor\P_1}{CC}{\Q}{\eff|\eff'}$---serves to split cases on the initial states, allowing different weavings to be used for different circumstances, which is why there is no notion like alignment guards for the biprogram conditional. 
The obvious conjunction rule is sound.\fullstuff{\footnote{\fullstuff{That is,
from 
$\Phi \rHPflowtr{}{}{\P_0}{CC}{\Q_0}{\eff|\eff'}$ and 
$\Phi \rHPflowtr{}{}{\P_1}{CC}{\Q_1}{\eff|\eff'}$ infer
$\Phi \rHPflowtr{}{}{\P_0\land\P_1}{CC}{\Q_0\land\Q_1}{\eff|\eff'}$.
The hypotheses and frame conditions are unchanged.
}}}
It is useful for deriving other rules.
For example, we have this simple axiom for allocation:
\( \rHPflowtr{}{}{true}{ \syncbi{x:=\new{K}} }{\later(x\eqbi x)}{\wri{x},\rw{\lloc}} \).
Using conjunction, embedding, \fullstuff{the unary rule \rn{Alloc},} and framing, one can add postconditions 
like $\Agr\sing{x}\Img f$ and freshness of $x$.

A consequence of our design decisions is ``one-sided divergence'' of biprograms,
which comes into play with weaving.
For example, assuming $loop$ diverges, 
$\Splitbi{y:=0;z.f:=0}{loop;x:=0}$ assigns $z.f$ before diverging.
But it weaves 
to $\splitbi{y:=0}{loop}; \splitbi{z.f:=0}{x:=0}$ which never assigns $z.f$.
This biprogram's executions do not cover all executions of the underlying unary programs.
The phenomenon becomes a problem for code that can fault (e.g., if $z$ is null).
Were the correctness judgments to assert termination, 
this shortcoming would not be an issue, but in this paper we choose the simplicity of partial correctness. 
Rule \rn{rWeave} needs to be restricted to prevent one-sided divergence of the premise biprogram $DD$ from states where $CC$ in the conclusion terminates.
For simplicity in this paper we assume given a termination check:
$terminates(P,C)$ means that $C$ faults or terminates normally, 
from any initial state satisfying $P$, 
This is about unary programs, so the condition can be discharged 
by standard means.\fullstuff{\footnote{\fullstuff{One can also think about a more complicated semantics for biprograms, in which splits take steps on alternating sides.  But this requires to augment configurations with some kind of scheduler state and  would slightly complicate some soundness proofs, so we leave that to future work.}}}

The relational frame rule 
is a straightforward extension of the unary frame rule.  
From a judgment 
$\Phi \rHPflowtr{}{}{\P}{CC}{\Q}{\eff|\eff'}$ it infers
$\Phi \rHPflowtr{}{}{\P\land \R}{CC}{\Q\land \R}{\eff|\eff'}$
provided that $\R$ is framed by read effects (on the left and right) that are disjoint from the write effects in $\eff|\eff'$.

To prove a judgment 
$\Phi \rHPflowtr{}{}{\Q}{\whilecbiA{E\smallSplitSym E'}{\P\smallSplitSym \P'}{CC}}{\Q}{\eff,\ftpt(E)|\eff',\ftpt(E')}$,
the rule 
has three main premises:\fullstuff{\footnote{\fullstuff{The syntactic footprint $\ftpt(E)$ is 
described in the Appendix and fully defined in~\cite{RegLogJrnI}.}}}
$\Phi \rHPflowtr{}{}{\Q\land\P\land\leftF{E}}{\splitbi{\Left{CC}}{\skipc}}{\Q}{\eff|\,}$
for left-only execution of the body,
$\Phi \rHPflowtr{}{}{\Q\land\P'\land\rightF{E'}}{\splitbi{\skipc}{\Right{CC}}}{\Q}{\,|\eff'}$ for 
right-only, and 
$\Phi \rHPflowtr{}{}{\Q\land\neg\P\land\neg\P'\land\leftF{E}\land\rightF{E'}}{CC}{\Q}{\eff|\eff'}$ 
for aligned execution.
A side condition requires that the invariant $\Q$ implies 
these cases are exhaustive:
$\Q\imp E\eqbi E' \lorbi (\P\land\leftF{E}) \lorbi (\P'\land\rightF{E'})$.
Additional side conditions require the effects to be self-immune, just as in unary RL~\cite{RegLogJrnI,BanerjeeNN15}.
Finally, the formulas  $\later\P\imp\P$ and $\later\P'\imp\P'$ must be valid; this 
says the alignment guards are refperm-independent, which is needed because refperms are part of the semantics of judgments but are not 
part of the semantics of biprograms.

The above rule is compatible with weaving a loop body, as in (\ref{eq:C1woven}).
The left and right projections $\Left{CC}$ and $\Right{CC}$ undo the weaving and take care of unaligned iterations.

There are many other valid and useful rules.
Explicit frame conditions are convenient, both in tools and in a logic,
in part because they compose in simple ways.
This may lose precision, but that can be overcome using postconditions 
to express, e.g., that $x:=x$ does not observably write $x$.
This is addressed, in unary RL, by a rule to ``mask'' write effects~\cite{RegLogJrnI}.
Similarly, the relational logic supports a rule to mask read effects.
There is a rule of transitivity along these lines:
$\rflowtrf{\P}{\splitbi{B}{C}}{\Q}$
and $\rflowtrf{\R}{\splitbi{C}{D}}{\S}$
infer 
$\rflowtrf{\P;\R}{\splitbi{B}{D}}{\Q;\S}$
where $(;)$ denotes composition of relations.  
A special case is where the pre-relations (resp.\ post-relations) are the same, transitive, relation.
The rule needs to take care about termination of $C$.



\section{Related work}\label{sec:related}

Benton~\cite{Benton:popl04} introduced relational Hoare logic, around the same time that Yang was developing relational separation logic~\cite{Yang07relsep}.
\fullstuff{Their works emphasize the effectiveness and flexibility of relational reasoning using ordinary extensional program semantics.} 
Benton's logic does not encompass the heap. Yang's does; it features separating conjunction and a frame rule. 
\fullstuff{The fully nondeterministic allocator is used.}
Pointers are treated concretely in~\cite{Yang07relsep};
agreement means identical addresses, which suffices for
some low level C code.
Neither work includes procedures.
Beringer~\cite{Beringer11} reduces relational verification to unary verification 
via specifications and uses that technique to derive rules of a relational Hoare logic for programs including the heap (but not procedures).  
Whereas the logics of Benton, Yang, and others provide only rules for synchronized alignment of loops,
Beringer derives a rule that allows for unsynchronized (``dissonant'') iterations;
our alignment guards are similar to side conditions of that rule.
RHTT~\cite{NanevskiBG13} implements a relational program logic in dependent type theory (Coq). 
The work focuses on applications to information flow. 
It handles dynamically allocated mutable state and procedures, 
and both similar and dissimilar control structures.
Like the other relational logics it does not feature frame conditions.
RHTT is the only prior relational logic to include both the heap and procedures,
and the only one to have a procedure linking rule.
It is also the only one to address any form of encapsulation;
it does so using abstract predicates, as opposed to hiding~\cite{RegLogJrnII,OHearnYangReynoldsInfoToplas}.

Several works investigate construction of product programs that encode 
nontrivial choices of alignment~\cite{ReynoldsCraft,TerauchiAiken,ZaksP08,BartheCK11,BartheCK13,BartheCK16}.
In particular, our weaving relation was inspired by~\cite{BartheCK11,BartheCK16} 
which address programs that differ in structure.
In contrast to the 2-safety properties for deterministic programs considered in this paper and most prior work,
Barthe et al.~\cite{BartheCK13} handle properties of the form 
``for all traces \ldots there exists a trace \ldots'' which are harder to work with but which encompass notions of refinement and continuity.
Relational specifications of procedures are used in a series of papers by Barthe et al. (e.g.,\cite{BartheKOB13}) for computer-aided cryptographic proofs.
Sousa and Dillig~\cite{SousaD16} implement a logic that encompasses $k$-ary relations, e.g.,
the 3-safety property that a binary method is computing a transitive relation; their verification
algorithm is based on an implicit product construction.
None of these works address the heap or the linking of procedure implementations.
\fullstuff{(Although the prototype implementation~\cite{SousaD16} does.)}
Several works show that syntactic heuristics can often find good 
weavings in the case of similarly-structured programs not involving the 
heap~\cite{KovacsSF13,MuellerKS15,SousaD16}. 
Mueller et al.~\cite{MuellerKS15} use a form of product program and a relational logic 
to prove correctness of a static analysis for dependency, including procedures but no heap.

Works on translation validation and conditional equivalence checking use verification conditions (VCs)
with implicit or explicit product constructions~\cite{ZaksP08,ZuckPGBFH05}.
Godlin and Strichman formulate and prove soundness of rules for proving equivalence of 
programs with similar control structure~\cite{GodlinS08}.
They use one of the rules to devise an algorithm for VCs using uninterpreted functions to encode equivalence of called procedures,
which has been implemented in two prototype tools for equivalence checking~\cite{GodlinS13}.
(Pointer structures are limited to trees, i.e., no sharing.)
Hawblitzel et al.~\cite{HawblitzelKLR13} 
and Lahiri et al.~\cite{LahiriMSH13}
use relational procedure summaries
for intra- and inter-procedural reasoning
about program transformations.
The heap is modeled by maps.
These and related works report good experimental results using SMT or SAT solvers
to discharge VCs.
Felsing et al.~\cite{FelsingGKRU14} use Horn constraint solving to infer coupling relations and relational procedure summaries, which works well for similarly structured programs; they do not deal with the heap.
The purpose of our logic is not to supplant VC-based tools approaches but rather to provide a foundation for them. 
Our biprograms and relational assertions are easily translated to SMT-based back ends like Boogie and Why3.

\nocite{LeinoMuellerESOP08}


Amtoft et al.~\cite{AmtoftBB06} introduce a logic for information flow in object-based programs,
using abstract locations to specify agreements in the heap.
It was proposed in~\cite{BNR08secpriv} to extend this approach to more general relational specifications, for fine-grained declassification policies.
Banerjee et al.~\cite{BNR08} 
showed how region-based reasoning including a frame rule can be encoded, using 
ghost code, with standard FOL assertions instead of an ancillary notion of abstract region.
This evolved to the logic in Section~\ref{sec:relLogic}.

Relational properties have been considered in the context of separation logic:~\cite{BirkedalY08rel} and~\cite{ThamsborgBY12} both give relational interpretations of unary separation logic that account for representation independence,
using second order framing~\cite{BirkedalY08rel} or abstract predicates~\cite{ThamsborgBY12}.
Extension of this work to a relational logic seems possible, 
but the semantics does not validate the rule of conjunction
so it may not be a good basis for verification tools.
Tools often rely heavily on splitting conjunctions in postconditions.



Ahmed et al.~\cite{AhmedDR09}  address representation independence for higher order code and code pointers, using a step-indexed relational model, and prove challenging instances of contextual equivalence.
Based on that work, Dreyer et al.~\cite{DreyerNRB10} formulate a relational modal logic for proving contextual equivalence for a language that has general recursive types and general ML-style references atop System F. 
The logic serves to abstract from details of semantics in ways likely to facilitate interactive proofs of interesting contextual equivalences, but it includes intensional 
atomic propositions about steps in the transition semantics of terms.
Whereas contextual equivalence means equivalent in all contexts, general relational logics 
can express equivalences conditioned on the initial state.  For example, 
the assignments $x:=y.f$ and $z.f:=w$ do not commute, in general, because their effects can overlap.
But they do commute under the precondition $y\neq z$. 
We can easily prove equivalence judgments such as 
\( \Splitbi{x:=y.f; z.f:=w}{z.f:=w; x:=y.f} :
   \Both{(y\neq z)} \land \Agr \sing{y}\Img f \land w \eqbi w
   \rspecSym
   x\eqbi x \land \Agr\sing{z}\Img f 
\). 
By contrast with~\cite{AhmedDR09,NanevskiBG13}, we do not rely on embedding in higher-order logic.

\fullstuff{%
Benton et al.~\cite{BentonHN14} give a region-based type and effect system that supports observational purity
and validates a number of equivalences that hold in virtue of effects alone.
The authors note that 
the semantics does not validate equivalences involving representation independence.
}

\section{Conclusion}\label{sec:discuss}

We provide a general relational logic that encompasses the heap and includes procedures. It handles both similarly- and differently-structured programs. 
We use small-step semantics with the goal to leverage, in future work, our prior work on SMT-friendly heap encapsulation~\cite{RosenbergBN12,RegLogJrnII,BanerjeeNN15}
for representation independence, which is not addressed in prior relational logics.\footnote{With the partial exception of~\cite{AhmedDR09}, see Section~\ref{sec:related}.
Although there has been some work on observational equivalence for higher order programs,
we are not aware of work dealing with general relational judgments for higher order programs.}

As articulated long ago by Hoare~\cite{Hoare:data:rep} but never fully formalized in a logic of programs, 
reasoning about change of data representation is  based on simulation relations on encapsulated state, which are necessarily preserved by client code in virtue of encapsulation.  
For functional correctness this corresponds to ``hiding'' of invariants on encapsulated data, i.e., not including the invariant in the specification used by a client.
O'Hearn et al.~\cite{OHearnYangReynoldsInfoToplas} formalize this as a hypothetical or second order framing rule (which has been adapted to RL~\cite{RegLogJrnII}).
In ongoing work, the logic presented here has been extended to address encapsulation and provides a relational second order frame rule which embodies Reynolds' abstraction theorem~\cite{Reynolds84}.
Whereas framing of invariants relies on write effects, framing of encapsulated relations also relies on read effects.
Our ongoing work also addresses observational purity, which is known to be
closely related to representation independence~\cite{Hoare:data:rep,obspureTCS}.

Although we can prove equivalence for loop tiling, some array-oriented loop optimizations seem to be out of reach of the logic as currently formulated.
Loop interchange changes matrix row to column order, reordering unboundedly many atomic assignments,
as does loop fusion/distribution.
Most prior work 
does not handle these examples;~\cite{ZuckPGBFH05} does handle them, 
with a non-syntactic proof rule that involves permutations on transition steps,
cf.~\cite{NamjoshiSinghania}.




\paragraph*{Acknowledgments}
Thanks to C{\'e}sar Kunz, Mounir Assaf, and Andrey Chudnov for helpful suggestions
and comments on previous versions of this paper.


\addcontentsline{toc}{section}{\protect\numberline{}{References}}

\bibliographystyle{plainurl}

\input{TRappendix}

\tableofcontents

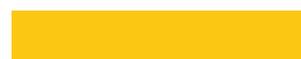
\end{document}


%% file: lipmacros.tex
\renewcommand{\emptyset}{\varnothing}
\renewcommand{\setminus}{\backslash} 

\definecolor{light-gray}{gray}{0.88}
\definecolor{dark-gray}{gray}{0.25}

\newcommand{\increment}{\mbox{\small \texttt{++}}}

\newcommand{\gassym}{\mathbin{{\color{blue} ::= }}} 
\newcommand{\gmid}{\mathrel{{\color{blue}\mid}}}

\newcommand{\dt}[1]{\index{#1}\textbf{\emph{#1}}} 

\newcommand{\proves}{\vdash}

\newcommand{\powerset}{\mathbb{P}}
\newcommand{\powersetne}{\mathbb{P}^\text{ne}}

\newcommand{\aftqua}{.\:}
\newcommand{\all}[2]{\forall #1 \aftqua #2}
\newcommand{\some}[2]{\exists #1 \aftqua #2}

\newcommand{\unionC}{\mathbin{\mbox{\small$\!\cup\!$}}} 
\newcommand{\union}{\mathbin{\mbox{\small$\cup$}}}
\newcommand{\intersect}{\mathbin{\mbox{\small$\cap$}}}
\newcommand{\imp}{\Rightarrow}

\newcommand{\eqdef}{\mathrel{\widehat{=}}}

\newcommand{\rn}[1]{\textsc{#1}} 

\newcommand{ \labf }[2]{\begin{trivlist}\item[$\;\;#1\eqdef$]$#2$ \end{trivlist}}

\newcommand{\keyw}[1]{\ensuremath{\mathsf{#1}}} 
\newcommand{\Region}{\keyw{rgn}}
\newcommand{\sing}[1]{\{#1\}} 

\newcommand{\NEW}{\keyw{new}}

\newcommand{\WHILE}{\keyw{while}}
\newcommand{\IF}{\keyw{if}}
\newcommand{\FI}{\keyw{fi}}
\newcommand{\DO}{\keyw{do}}
\newcommand{\OD}{\keyw{od}}

\newcommand{\THEN}{\keyw{then}}

\newcommand{\new}[1]{\NEW\;#1} 
\newcommand{\skipc}{\keyw{skip}}

\newcommand{\Endlet}{\keyw{elet}} 

\newcommand{\seqc}[2]{\ensuremath{{#1}\:;{#2}}} 
\newcommand{\ifc}[3]{\keyw{if}\ {#1}\ \keyw{then}\ {#2}\ \keyw{else}\ {#3}}
\newcommand{\whilec}[2]{\keyw{while}\ {#1}\ \keyw{do}\ {#2}}
\newcommand{\BE}{\mathbin{=}}
\newcommand{\letcom}[3]{\keyw{let}~#1 \BE #2~\keyw{in}~#3} 
\newcommand{\letcomMulti}[2]{\keyw{let}~#1~\keyw{in}~#2} 

\newcommand{\INT}{\keyw{int}}
\newcommand{\BOOL}{\keyw{bool}}

\newcommand{\NULL}{\keyw{null}}

\newcommand{\ol}{\overline}

\newcommand{\uequiv}{\mathrel{\cong}} 
\newcommand{\allowTo}{\mathord{\to}} 
\newcommand{\allowDep}{\mathord{\Rightarrow}} 
\newcommand{\successorTo}{\hookrightarrow}
\newcommand{\trans}[1]{\stackrel{#1}{\longmapsto}} 

\newcommand{\tranStar}[1]{\mathbin{\stackrel{#1}{\longmapsto}{\!\!}^*}}

\newcommand{\config}[2]{\langle #1,\: #2\rangle}
\newcommand{\configm}[3]{\langle #1,\: #2,\: #3\rangle} 
\newcommand{\configc}[1]{\langle #1 \rangle} 

\newcommand{\mathconst}[1]{\mbox{\textsl{#1}}} 

\newcommand{\agree}{\mathconst{Agree}} 
\newcommand{\rlocs}{\mathconst{rlocs}} 
\newcommand{\wlocs}{\mathconst{wlocs}} 
 
\newcommand{\Lagree}{\mathconst{Lagree}} 
\newcommand{\freshRefs}{\mathconst{freshRefs}} 
\newcommand{\freshLocs}{\mathconst{freshLocs}}

\newcommand{\fields}{\mathconst{Fields}} 
 
\newcommand{\fresh}{\mathconst{Fresh}}

\newcommand{\update}[3]{[#1\, |\, #2:\, #3]} 
\newcommand{\scol}{\mathord{:}} 
\newcommand{\extend}[3]{[#1 \mathord{+} #2\scol\, #3]} 
\newcommand{\drop}[2]{#1\mathbin{\!\upharpoonright\!}#2} 

\newcommand{\type}{\mathconst{Type}} 
\newcommand{\dom}{\mathconst{dom}\,}

\newcommand{\Default}[1]{\mathconst{default}(#1)}
\newcommand{\semNull}{\mathconst{null}} 
\newcommand{\Ref}{\mathconst{Ref}}
\newcommand{\True}{\keyw{true}}
\newcommand{\False}{\keyw{false}}

\newcommand{\Rall}[1][]{\mathit{R}}

\newcommand{\Emp}{\emptyset} 
\newcommand{\lloc}{\keyw{alloc}} 

\newcommand{\disj}[2]{#1\mathbin{\mbox{\#}}#2} 
\newcommand{\ind}[2]{#1 \indSymbol #2}

\newcommand{\indSymbol}{\mathbin{\cdot\mbox{\small\textbf{/}}.}}

\newcommand{\ftpt}{\mathconst{ftpt}} 
\newcommand{\fra}[2]{#1\mathrel{\keyw{frm}}#2} 

\newcommand{\flowtr}[4]{#2:\: #1\leadsto #3\:[#4]}
\newcommand{\sflowtr}[4]{#2:\: #1\leadsto #3\:[#4]}
\newcommand{\flowtrTy}[3]{#1\leadsto #2\:[#3]}
\newcommand{\flowty}[3]{#1\leadsto #2\:[#3]} 
\newcommand{\rflowty}[3]{#1\rspecSym #2\:[#3]} 

\newcommand{\rflowtr}[4]{#2:\: #1\rspecSym #3\:[#4]}
\newcommand{\rflowtrf}[3]{#2:\: #1\rspecSym #3} 

\newcommand{\rHPflowtr}[6]{\proves^{#1}_{#2}\rflowtr{#3}{#4}{#5}{#6}} 
\newcommand{\rHVflowtr}[6]{\models^{#1}_{#2}\rflowtr{#3}{#4}{#5}{#6}}

\newcommand{\HPflowtr}[6]{\proves^{#1}_{#2}\flowtr{#3}{#4}{#5}{#6}} 
\newcommand{\HVflowtr}[6]{\models^{#1}_{#2}\flowtr{#3}{#4}{#5}{#6}} 

\newcommand{\Img}{\mbox{\large\textbf{`}}}
\newcommand{\allfields}{\mathsf{any}} 

\newcommand{\eff}{\varepsilon} 
 
\newcommand{\effe}{\eta}
\newcommand{\effd}{\delta} 

\newcommand{\wri}[1]{\keyw{wr}\,#1} 
\newcommand{\rd}[1]{\keyw{rd}\,#1} 
\newcommand{\rw}[1]{\keyw{rw}\,#1} 

\newcommand{\ldb}{[\![} 
\newcommand{\rdb}{]\!]} 
\newcommand{\means}[1]{\ldb {#1}\rdb}

\renewcommand{\phi}{\varphi}

\newcommand{\rprel}[2]{\Rprel{\pi}{#1}{#2}}
\newcommand{\Rprel}[3]{#2\stackrel{#1}{\sim}#3}

\newcommand{\configr}[5]{\langle #1,\: #2|#3,\:  #4|#5\rangle}

\newcommand{\syncbi}[1]{\lfloor #1 \rfloor}
\newcommand{\Syncbi}[1]{\llfloor #1 \rrfloor} 
\newcommand{\splitbi}[2]{(#1|#2)}
\newcommand{\Splitbi}[2]{(#1\mid#2)} 

\newcommand{\smallSplitSym}{\mbox{$|$}} 

\newcommand{\alignSepSym}{\mbox{\small$\bullet$}} 

\newcommand{\ifcbi}[3]{\keyw{if}\ {#1}\ \keyw{then}\ {#2}\ \keyw{else}\ {#3}}

\newcommand{\letcombi}[3]{\keyw{let}~#1 \BE #2~\keyw{in}~#3} 

\newcommand{\whilecbiA}[3]{\keyw{while}\ {#1}\ \alignSepSym\ {#2}\ \keyw{do}\ {#3}}

\newcommand{\Fault}{\lightning}
\newcommand{\PFault}{\Fault}
\newcommand{\AFault}{\Fault}
\newcommand{\AFPFault}{\Fault}

\newcommand{\biTrans}[2]{\stackrel[#1]{#2}{\Longmapsto}}  
\newcommand{\biTranStar}[2]{\mathbin{\stackrel[#1]{#2}{\Longmapsto}{\!\!}^*}}

\def\rightharpoonupfill{$\mathsurround=0pt \mathord- \mkern-6mu
  \cleaders\hbox{$\mkern-2mu \mathord- \mkern-2mu$}\hfill
  \mkern-6mu \mathord\rightharpoonup$}

\def\leftharpoonupfill{$\mathsurround=0pt \mathord\leftharpoonup \mkern-6mu
  \cleaders\hbox{$\mkern-2mu \mathord- \mkern-2mu$}\hfill
  \mkern-6mu \mathord-$}

\def\overleftharpoonup#1{\vbox{\ialign{##\crcr
      \leftharpoonupfill\crcr\noalign{\kern-1pt\nointerlineskip}
      $\hfil\displaystyle{#1}\hfil$\crcr}}}

\def\overrightharpoonup#1{\vbox{\ialign{##\crcr
      \rightharpoonupfill\crcr\noalign{\kern-1pt\nointerlineskip}
      $\hfil\displaystyle{#1}\hfil$\crcr}}}

\def\overleftrightharpoonup#1{\vbox{\ialign{##\crcr
      \leftharpoonupfill\hspace*{-.7em}\rightharpoonupfill\crcr\noalign{\kern-1pt\nointerlineskip}
      $\hfil\displaystyle{#1}\hfil$\crcr}}}

\newcommand{\Left}[1]{\overleftharpoonup{#1}} 
\newcommand{\Right}[1]{\overrightharpoonup{#1}}

\newcommand{\written}{\mathconst{written}}

\newcommand{\rspecSym}{\ensuremath{\mathrel{\mbox{\footnotesize$\approx\!>$}}}}

\newcommand{\una}[1]{\lbag#1\rbag} 

\renewcommand{\P}{\mathcal{P}}  
\renewcommand{\S}{\mathcal{S}} 
\renewcommand{\L}{\mathcal{L}}  
\newcommand{\Q}{\mathcal{Q}} 
\newcommand{\R}{\mathcal{R}} 

\newcommand{\Z}{\mathbb{Z}} 

\newcommand{\weave}{\hookrightarrow} 

\newcommand{\Agr}{\ensuremath{\mathbb{A}}}
\newcommand{\lorbi}{\lor}
\newcommand{\eqbi}{\mathrel{\ddot{=}}}
\newcommand{\later}{\diamond}
\newcommand{\always}{\mathord{\mbox{\small$\Box$}}}
\newcommand{\leftF}[1]{\triangleleft #1} 
\newcommand{\rightF}[1]{\triangleright #1} 
\newcommand{\Both}[1]{\ensuremath{\mathbb{B}} #1}

\newcommand{\hbra}{
  \hbox to \columnwidth{\vrule width0.3mm height 1.8mm depth-0.3mm
    \leaders\hrule height1.8mm depth-1.5mm\hfill
    \vrule width0.3mm height 1.8mm depth-0.3mm}}
\newcommand{\hket}{
  \hbox to \columnwidth{\vrule width0.3mm height1.5mm
    \leaders\hrule height0.3mm\hfill
    \vrule width0.3mm height1.5mm}}

\newcommand{\ratio}{.35}

\newenvironment{tdisplay}[1]{\medskip
    \noindent\textbf{\normalsize #1}\\[-.3ex]
    \hbra\\[-.4ex]
  }{\par\hket
  \medskip}


%% file: TRappendix.tex
\clearpage

\appendix

\section{Semantics of unary programs and their correctness judgments}\label{app:unary}

A typing context, $\Gamma$, maps variables to types.  (Types are in Fig.~\ref{fig:bnf}.)
A \dt{$\Gamma$-state} is comprised of a heap and a type-respecting assignment of values to the variables in $\Gamma$, which always includes the special variable $\lloc$, built into the semantics, that is not allowed to be assigned in code.  
Its value is the set of allocated references.  
It appears in frame conditions of code that allocates,
a detail that is glossed over in Sec.~\ref{sec:overview}.
A state must be well formed in the sense that there are no dangling references.
In particular, the value of a region expression is a set of allocated references, possibly also including $\semNull$.
We write $\sigma(x)$ to look up the value of $x$ in state $\sigma$, 
$\sigma(o.f)$ to look up field $f$ of reference $o$, 
$\sigma(F)$ for the value of expression $F$,
and $\means{\Gamma}$ for the set of $\Gamma$-states.

The transition semantics uses configurations $\configm{C}{\sigma}{\mu}$ where 
$\mu$ is an environment that maps procedure names to commands.
(The control state $C$ encodes a stack of continuations as a single command, using scope endmarkers for $\keyw{var}$ and $\keyw{let}$.
Nothing is needed to mark the end of a procedure call, as procedures have neither parameters nor returns.)
We work with typed configurations, and typed correctness judgments, but gloss over typing in this paper (see~\cite{RegLogJrnII}).
The transition semantics is standard,
except that we aim for reasoning about programs under hypotheses, i.e., procedure specifications, as explained in due course.
  
The heap is unbounded.  The command  
  $x:=\new{K}$ allocates a fresh reference and maps it to an object of type $K$
  initialized  with 0-equivalent values.  In order to model real allocators, which may
  depend on state not visible at the language level, we assume an arbitrary choice
  function for fresh references, which may be, but need not be, nondeterministic.

The semantics of formulas is standard.
The points-to relation is defined by $\sigma\models x.f=E$ iff 
$\sigma(x)\neq\semNull$ and $\sigma(\sigma(x).f)=\sigma(E)$.
Quantifiers for reference types range over allocated non-null references:
$\sigma\models \all{x:K}{P}$ iff $\extend{\sigma}{x}{o}\models P$ for all $o\in\sigma(\lloc)$
of type $K$.
(The notation indicates extending $\sigma$ with $x$ mapped to $o$.)

The meaning of a correctness judgment 
is defined in terms of executions from initial configurations where the environment is empty (written $\_$).
Recall that in $C$ there may be environment calls to procedures bound by $\keyw{let}$ in $C$ and there may also be context calls to procedures in  a hypothesis context $\Phi$.
In the transition semantics, context calls take a single step to an outcome in accord with the specification:
if the pre-state satisfies the precondition then the post-state satisfies the postcondition, 
and otherwise the outcome is fault ($\Fault$).
In~\cite{RegLogJrnII} and~\cite{OHearnYangReynoldsInfoToplas},
this kind of semantics is defined in terms  of a single transition relation for the procedure, which 
encodes under-specification by nondeterminacy.
Here, a key design choice is to avoid nondeterminacy, to cater for simple semantics of relational properties.
Following~\cite{BanerjeeNN15}, this is achieved by semantics of 
correctness judgments
in terms of all ``interpretations'' of the hypotheses, each interpretation being  
deterministic up to renaming of references.
This is captured in the $\forall\forall$ semantics of read effects
(item (c) of Context Interpretation, below).\footnote{We have to deal with renaming of references in any case, even if the allocator is deterministic, to handle properties like noninterference for a program that allocates differently depending on secrets, or two versions of an algorithm using different pointer structures.}

A \dt{location} is a variable $x$ or a pair $o.f$ of a reference $o$ and field $f$.
Define $\rlocs(\sigma,\eff)$, the locations designated in $\sigma$ by read effects of $\eff$, by
\( \rlocs(\sigma,\eff) = 
    \{ x \mid \mbox{$\eff$ contains $\rd{x}$} \} 
   \union \{ o.f \mid \mbox{$\eff$ contains $\rd{G\Img f}$ with
                     $o\in \sigma(G)$} \} 
\).
Define the $\wlocs(\sigma,\eff)$ the same way but for write effects.
Say $\tau$ \dt{can succeed} $\sigma$, written $\sigma\successorTo\tau$, provided 
$\sigma(\lloc)\subseteq\tau(\lloc)$ and
$\type(o, \sigma) = \type(o, \tau)$ for all $o\in\sigma(\lloc)$.
Define $\written(\sigma,\tau)$ to be 
$\{ x \mid \sigma(x)\neq\tau(x) \} \union \{ o.f \mid \sigma(o.f)\neq\tau(o.f) \}$.
Say $\eff$ \dt{allows change from} $\sigma$ \emph{to} $\tau$, 
\index{$\sigma \allowTo \tau \models \eff$}
written $ \sigma \allowTo \tau \models \eff $,
iff $\sigma\successorTo\tau$ and $\written(\sigma,\tau)\subseteq \wlocs(\sigma,\eff)$.

The semantics of read effects is more involved.
Let $\pi$ range over \dt{partial bijections} on $\Ref\setminus\{\semNull\}$.
Write $\pi(p)=p'$ to say that $\pi$ is defined on $p$ and has value $p'$. 
A \dt{refperm} from $\sigma$ to $\sigma'$ is a partial bijection $\pi$ such that
$\dom(\pi)\subseteq \sigma(\lloc)$,
$rng(\pi)\subseteq \sigma'(\lloc)$, and 
$\pi(p)=p'$ implies $\type(p,\sigma)=\type(p',\sigma')$ for all proper references $p,p'$.
For references $o,o'$ define $\rprel{o}{o'}$ to mean 
$o=\semNull=o'$ or $\pi(o)=o'$.
Extend $\stackrel{\pi}{\sim}$ to a relation on integers by $\rprel{i}{j}$ iff $i=j$.  For reference sets $X,Y$, define $\rprel{X}{Y} \;\mbox{ iff }\; \pi(X)=  Y$,
where $\pi(X)$ is the direct image of $X$.
%
%
For a set $W$ of locations, define 
$\Lagree(\sigma,\sigma',\pi,W)$
iff 
\[\begin{array}[t]{l}
\all{x\in W}{\: \rprel{\sigma(x)}{\sigma'(x)} } \; \land \;
\all{(p.f)\in W}{\: p\in dom(\pi) \:\land\: \rprel{ \sigma(p.f) }{ \sigma'(\pi(p).f) }} 
\end{array}
\]\index{$\Lagree$}
%
%
Say that $\sigma$ and $\sigma'$ \dt{agree on $\eff$ modulo $\pi$},
written $\agree(\sigma, \sigma', \eff, \pi)$,
iff
$\Lagree(\sigma,\sigma',\pi,\rlocs(\sigma,\eff))$.


Note that  $\agree(\sigma,\tau,\rd{G}\Img f,\pi)$ 
implies $\sigma(G)\subseteq \dom(\pi)$ but does not imply 
$\tau(G)\subseteq rng(\pi)$ or $\rprel{\sigma(G)}{\tau(G)}$.

The next definitions are the basis for the semantics of read effects,
which is a relational property of two initial states $\sigma,\sigma'$ and two final states $\tau,\tau'$.  

\begin{tdisplay}{Allowed dependence \hfill
$\sigma,\sigma'\allowDep\tau,\tau' \models \eff$} 
\vspace*{-2.5ex}
\[ 
\hspace*{-1ex}
\begin{array}{l}
\freshRefs(\sigma,\tau) \eqdef \tau(\lloc)\setminus \sigma(\lloc) \\

\freshLocs(\sigma,\tau)
\eqdef \{ p.f | p\in\freshRefs(\sigma,\tau),
           f\in\fields(\type(p,\tau)) \}
\end{array}
\]

\smallskip

\noindent Say $\eff$ \dt{allows dependence} from $\tau,\tau'$ to $\sigma,\sigma'$, written
$\sigma,\sigma'\allowDep\tau,\tau' \models \eff$, 
iff 
for all $\pi$ if $\agree(\sigma,\sigma',\eff,\pi)$
then there is $\rho\supseteq\pi$ such that 
$\Lagree(\tau,\tau',\rho, \freshLocs(\sigma,\tau)\union \written(\sigma,\tau))$.
\end{tdisplay}


An interpretation returns a non-empty set of outcomes from each initial state (notation $\powersetne$). 

\begin{tdisplay}{Context interpretation \hfill $\phi$ for $\Phi$}
For $\Phi$ well formed in $\Gamma$, 
a \dt{$\Phi$-interpretation}
is a function $\phi$ with $\dom\phi=\dom\Phi$ and for each $m: \flowtrTy{P}{Q}{\eff}$ in $\Phi$,
we have that 
$\phi(m)$ is a function $\means{\Gamma}\to\powersetne(\means{\Gamma}\union\{\Fault\})$
such that for all $\sigma \in \means{\Gamma}$
we have
\begin{list}{}{}
\item[(a)] $\Fault\in\phi(m)(\sigma)$ iff $\sigma\not\models P$,
and also $\Fault\in\phi(m)(\sigma)$ implies $\phi(m)(\sigma) = \{\Fault\}$.
\item[(b)] For all $\tau \in \phi(m)(\sigma)$,
if $\sigma\models P$
then $\tau\models Q$ 
and $\sigma\allowTo\tau\models \eff$.
\item[(c)] For all $\tau,\sigma',\tau'$, 
if $\sigma\models P$ 
and $\sigma'\models P$
and $\tau \in \phi(m)(\sigma)$ and
$\tau' \in \phi(m)(\sigma')$, then
$\sigma,\sigma'\allowDep\tau,\tau'\models\eff$.
\end{list}
\vspace*{-3ex}
\end{tdisplay}
Owing to the second condition in (a),
dubbed \dt{fault determinacy},
we could as well choose to treat $\phi(m)$ as a function with 
codomain $\powersetne(\means{\Gamma})\union\{\Fault\}$, but the chosen formulation slightly streamlines some definitions.

The transition relation $\trans{\phi}$ depends on an interpretation $\phi$.
Transitions act on configurations where the environment $\mu$ has procedures distinct from those of $\phi$.
Aside from the use of interpretations, the definition is mostly standard (and omitted).
We assume $\fresh$ is a function such that, for any $\sigma$,
$\fresh(\sigma)$ a non-empty set of non-null references that are not 
in $\sigma(\lloc)$.

\begin{tdisplay}{Selected transition rules \hfill $\trans{\phi}$}
\vspace*{-3ex}
\begin{small}
\begin{mathpar}
   \inferrule{ \mu(m) = C \\
   }{ \configm{ m() }{\sigma}{\mu} 
      \trans{\phi}
      \configm{C}{\sigma}{\mu}
    }


\and
 \inferrule
   { \tau\in\phi(m)(\sigma) }
   { \configm{m()}{\sigma}{\mu} \trans{\phi} \configm{\skipc}{\tau}{\mu} }

\and
   \inferrule
   { \Fault\in\phi(m)(\sigma) }
   { \configm{m()}{\sigma}{\mu} \trans{\phi} \Fault}

\and
  \inferrule{ o \in \fresh(\sigma) \\
              \fields(K) \:=\:\ol{f}:\ol{T} \\
              \sigma_1 = \mbox{``$\sigma$ with $o$ added to heap, with type $K$ and default fields''}
            }
            { \configm{x:=\new{K}}{\sigma}{\mu} \trans{\phi}
              \configm{\skipc}{\update{\sigma_1}{x}{o} }{\mu}}

\and
\inferrule{} 
{  \configm{ \letcom{m()}{B}{C} }{\sigma}{\mu} 
      \trans{\phi} \configm{C;\Endlet(m)\,}{\sigma}{\extend{\mu}{m}{B}}
}

\and
 \configm{ \Endlet(m) }{\sigma}{\mu} \trans{\phi} \configm{ \skipc}{\sigma}{\drop{\mu}{m}}

\end{mathpar}
\end{small}
\vspace*{-3ex}
\end{tdisplay}


The correctness judgment gives a modular form of partial correctness.  
The avoidance of faults says not only that there are no null dereferences but more importantly that no context procedure is called outside its specified precondition.

\begin{tdisplay}{Valid correctness judgment
\hfill $\Phi\HVflowtr{}{}{P}{C}{Q}{\eff}$ 
}
The judgment is \dt{valid} 
iff the following conditions hold for 
all $\Phi$-interpretations $\phi$ 
and all $\Gamma$-states $\sigma$ such that  $\sigma\models P$.
\begin{list}{}{}
\item[\quad (Safety)] It is not the case that $\configm{C}{\sigma}{\_} \tranStar{\phi} \,\Fault$.
\end{list}
And for every $\tau$ with $\configm{C}{\sigma}{\_} \tranStar{\phi} \configm{\skipc}{\tau}{\_}$
we have
\smallskip
\begin{list}{}{}
\item[\quad (Post)]
$\tau \models Q$
\item[\quad (Write Effect)] $\sigma\allowTo\tau\models \eff$
\item[\quad (Read Effect)]
For all $\sigma',\tau'$
if 
$\configm{C}{\sigma'}{\_} \tranStar{\phi} \configm{\skipc}{\tau'}{\_}$
and $\sigma'\models P$
then $\sigma,\sigma'\allowDep\tau,\tau'\models\eff$
\end{list}
\vspace*{-2ex}
\end{tdisplay}

Selected proof rules appear in Fig.~\ref{fig:unaryRules}.
We proceed to some notions used in the rules, starting with read effects of formulas.

\begin{tdisplay}{Framing of formulas \hfill $P \models \fra{\effe}{Q}$}
$P \models \fra{\effe}{Q} \; $ iff 
for all $\sigma, \sigma', \pi$,
if $\agree(\sigma, \sigma', \effe, \pi)$
and $\sigma \models P \land Q$ then 
$\sigma' \models Q$
\end{tdisplay}

For atomic formulas, read effects can be computed syntactically by function $\ftpt$.
Two clauses of the definition~\cite{RegLogJrnI} 
are $\ftpt(x) = \rd{x}$ and $\ftpt(x.f = F) = \rd{x},\rd{\sing{x}\Img f}, \ftpt(F)$.
The basic lemma is that 
$\agree(\sigma, \sigma', \ftpt(F), \pi)$ implies 
$\rprel{\sigma(F)}{\sigma'(F)}$.

To express region disjointness we use a syntactic function $\indSymbol$ defined by structural recursion on effects (see~\cite{RegLogJrnI}).
Please note that $\indSymbol$ is not syntax in the logic; it's a function in the metalanguage that is used to obtain formulas from effects.
For example, $\ind{r\Img nxt}{r\Img val}$ is the formula $true$ and 
$\ind{r\Img nxt}{s\Img nxt}$ is the formula $r\intersect s \subseteq \{\semNull\}$.
The key lemma is that the formula $\ind{\eff}{\effe}$ 
holds in a state $\sigma$ iff $\rlocs(\sigma,\eff)\intersect \wlocs(\sigma,\effe)=\emptyset$.

The \dt{subeffect} judgment $P\models \eff \leq \effe$ holds iff 
$\rlocs(\sigma,\eff)\subseteq\rlocs(\sigma,\eta)$ and 
$\wlocs(\sigma,\eff)\subseteq\wlocs(\sigma,\eta)$ for all $\sigma$ with $\sigma \models P$.
The key lemma about subeffects is that 
if $\sigma \allowTo \tau \models \eff$
and  $\sigma\models \ind{\effe}{\eff}$ 
and $P\models \eff \leq \effe$ 
and $\sigma\models P$
then $\agree(\sigma, \tau, \effe, id)$ where $id$ is the identity on $\sigma(\lloc)$.

\begin{figure}[t]
\begin{small}
\begin{mathpar}
\inferrule*[left=Frame]
{ \Phi\HPflowtr{}{}{P}{C}{Q}{\eff} \\
  P \models \fra{\effe}{R} \\
  P\land R \imp \ind{\effe}{\eff}
}{ 
\Phi\HPflowtr{}{}{P\land R}{C}{Q\land R}{\eff}
} 

\and
\inferrule*[left=FieldUpd]
{y\not\equiv x }
{ \HPflowtr{}{}{
  x \neq \NULL }{x.f := y\;}{\,x.f =y}{\wri{\sing{x}\Img f},\rd{x},\rd{y}}}

\and
\inferrule*[left=Alloc]
{ \fields(K) = \ol{f}:\ol{T} \\
  \Default{\ol{T}} = \ol{v} 
}
{
  \mbox{\(\begin{array}[t]{ll}
  \proves x:=\new{K} : ~ 
             \lloc = g \leadsto \lloc = g\union\{x\} \land x.\ol{f}=\ol{v}
   \:[ \wri{x}, \rw{\lloc} ]
  \end{array}
  \)}
}

\and
\inferrule*[left=Seq]  
{ 
\Phi\HPflowtr{}{}{P}{C_1}{P_1}{\eff_1}\\
\Phi\HPflowtr{}{}{P_1}{C_2}{Q}{\eff_2,\rw{H\Img\ol{f}}}\\
P_1\Rightarrow H\#g\\
\eff_2 \mbox{ is } P/\eff_1\mbox{-immune}\\
\wri{g}\not\in\eff_1
} 
{
\Phi \HPflowtr{}{}{P\wedge g=\lloc}{C_1;C_2}{Q}{\eff_1,\eff_2}
}

\end{mathpar}
\end{small}
\vspace*{-3ex}
\caption{Selected unary proof rules (from~\cite{RegLogJrnII,BanerjeeNN15}).
}\label{fig:unaryRules}
\end{figure}

Separator formulas are used in the notion of immunity, which amounts to framing for frame conditions.
Expression $G$ is \dt{$P/\eff$-immune}
iff this is valid:
\( P\imp \ind{\ftpt(G)}{\eff}  \).
Effect  $\effe$ is \dt{$P/\eff$-immune} iff 
$G$ is $P/\eff$-immune for every $G$ with 
$\wri{G\Img f}$ or $\rd{G\Img f}$ in $\effe$.

The key lemma about immunity is that 
if $\effe$ is $P/\eff$-immune,
$\sigma\models P$, and $\sigma\allowTo\tau \models\eff$,
then 
$\rlocs(\sigma,\effe) = \rlocs(\tau,\effe)$ and 
$\wlocs(\sigma,\effe) = \wlocs(\tau,\effe)$.

The preceding notions are concerned with protecting formulas and effects from the write effects of a command.
That is, framing and immunity are about preserving the value of an expression or formula from one control point to a later one.  To deal with read effects (and other relations), agreements also need to be preserved.  
To this end we use the following notion.
An effect $\eff$ has \dt{framed reads} provided that for every $\rd{G\Img f}$ in $\eff$, 
its footprint $\ftpt(G)$ is in $\eff$.
For example, with $r: \Region$ the effect $\rd{r\Img f}$ does not have framed reads, but it is a subeffect of $\rd{r\Img f},\rd{r}$ which does.
For $\eff$ that has framed reads, if $\agree(\sigma,\sigma',\eff,\pi)$
then $\rprel{\sigma(G)}{\sigma'(G)}$ for any $\rd{G\Img f}$ in $\eff$.
In addition, a kind of symmetry holds:
$\agree(\sigma,\sigma',\eff,\pi)$ implies
$\agree(\sigma',\sigma,\eff,\pi^{-1})$.
This property implies what we need for preservation of effects (see \cite{BanerjeeNN15} for details).

In this paper we assume without comment that all frame conditions in unary and relational judgments have framed reads.
(An alternative would be to change the semantics so that effects are interpreted 
in terms of their $\ftpt$-closure.)

\section{Relation formulas}

Relational correctness judgments are typed in a context of the form 
$\Gamma|\Gamma'$ comprised of contexts $\Gamma$ and $\Gamma'$ for the left and right sides.
For relation formulas, typing is reduced to typing of unary formulas:
$\Gamma|\Gamma'\proves \P \;$ iff $\; \Gamma\proves\Left{\P}\mbox{ and }\Gamma'\proves\Right{\P}$.
This refers to the following.
\begin{tdisplay}{Syntactic projections}
\[ \begin{array}{l@{\hspace{1.5ex}}l@{\hspace{1.5ex}}l@{\hspace{2em}}l@{\hspace{1.5ex}}l@{\hspace{1.5ex}}l} 
\Left{\leftF{P}} & \eqdef & P 
  & \Right{\leftF{P}} & \eqdef & \True \\
\Left{\rightF{P}} & \eqdef & \True
  & \Right{\rightF{P}} & \eqdef & P \\
\Left{\later\P} & \eqdef & \Left{\P}
  & \Right{\later\P} & \eqdef & \Right{\P} \\ 
\Left{F\eqbi F'} & \eqdef & (F=F)
  & \Right{F\eqbi F'} & \eqdef & (F'=F') \\  
\Left{\Agr G\Img f} & \eqdef & (G\Img f = G\Img f)
  & \Right{\Agr G\Img f} & \eqdef & (G\Img f = G\Img f) \\
\Left{\all{x\smallSplitSym x':K}{\P}}  & \eqdef & \all{x:K}{\Left{\P}}   
& \Right{\all{x\smallSplitSym x':K}{\P}} & \eqdef & \all{x':K}{\Right{\P}}  
\end{array}
\]
\end{tdisplay}
Next are various notions used in the semantics of the program logic, starting 
with read effects of formulas.

\begin{tdisplay}{Relation formula semantics \hfill $\sigma|\sigma'\models_\pi\P$}
\begin{small}
\(
\begin{array}{l@{\hspace*{1ex}}l@{\hspace*{1ex}}l}
\sigma|\sigma'\models_\pi\leftF{P} 
  & iff \; & \sigma\models P \\

\sigma|\sigma'\models_\pi\rightF{P} 
  & iff \; & \sigma'\models P \\

\sigma|\sigma'\models_\pi F\eqbi F' 
  & iff & \rprel{\sigma(F)}{\sigma'(F')} \\


\sigma|\sigma'\models_\pi \Agr G\Img f 
  & iff & \agree(\sigma,\sigma',\rd{G\Img f},\pi) \\ 


\sigma|\sigma'\models_\pi \later \P 
  & iff & \some{\rho}{\rho\supseteq\pi \mbox{ and } \sigma|\sigma'\models_\rho \P } \\

\sigma|\sigma'\models_\pi  \P\land\Q 
  & iff & \sigma|\sigma'\models_\pi  \P \mbox{ and } \sigma|\sigma'\models_\pi \Q \\

\sigma|\sigma'\models_\pi  \P\imp\Q 
  & iff & \sigma|\sigma'\models_\pi  \P \mbox{ implies } \sigma|\sigma'\models_\pi \Q  \\



\end{array}\)
\end{small}
\end{tdisplay}

The framing judgment is like the unary one:
$\P\models \fra{\effd|\effd'}{\Q}$ 
iff for all $\pi, \sigma, \sigma', \tau, \tau'$, 
if $\agree(\sigma, \tau, \effd, id)$
and $\agree(\sigma', \tau', \effd', id)$
and $\sigma|\sigma' \models_\pi^{\Gamma|\Gamma'} \P \land \Q$ then $\tau|\tau' \models_\pi^{\Gamma|\Gamma'} \Q$.

\section{Biprograms}\label{app:biprograms}

\begin{tdisplay}{Biprograms: syntactic projections
\hfill $\Left{CC},\Right{CC}$}
\(
\begin{array}{lll}
\Left{\configm{CC}{\sigma|\sigma'}{\mu|\mu'}} & = & \configm{\Left{CC}}{\sigma}{\mu} \\

\Left{\splitbi{C}{C'}} & = & C \\
\Left{\syncbi{A}} & = & A \\
\Left{\ifcbi{E\smallSplitSym E'}{BB}{CC}} & = & \ifc{E}{\Left{BB}}{\Left{CC}} \\
\Left{\whilecbiA{E\smallSplitSym E'}{\P\smallSplitSym \P'}{CC}} & = & \whilec{E}{\Left{CC}} \\
\Left{\seqc{BB}{CC}} & = & \seqc{\Left{BB}}{\, \Left{CC}} \\
\Left{\letcombi{m}{\splitbi{C}{C'}}{CC}} & = & \letcom{m}{C}{\Left{CC}}
\end{array}
\)
\\[.5ex]
We identify $(\skipc;C) \equiv C$, 
$(C;\skipc) \equiv C$, 
$\splitbi{\skipc}{\skipc} \equiv \syncbi{\skipc}$
and $\syncbi{\skipc};CC \equiv CC$.
Thus, for example, $\Left{\splitbi{\skipc}{x:=0};\syncbi{y:=0}}$ is $y:=0$.
\end{tdisplay}

Typing of biprograms can be defined in terms of these meta-operators, roughly as $\Gamma|\Gamma'\proves CC$ iff $\Gamma\proves\Left{CC}$ and $\Gamma'\proves\Right{CC}$. But the alignment guards $\P,\P'$ in \keyw{while} should also be typechecked (by evident rules).

Biprograms are given transition semantics, with configurations of the form 
$\configm{CC}{\sigma|\sigma'}{\mu|\mu'}$ that represent an aligned pair of unary configurations.
Environments are unchanged from unary semantics: $\mu$ and $\mu'$ map procedure names to commands, not to biprograms.
We lift the syntactic projections to configurations:
$\Left{\configm{CC}{\sigma|\sigma'}{\mu|\mu'}} 
= \configm{\Left{CC}}{\sigma}{\mu}$. 

\begin{figure}[t!]
\begin{small} 
\begin{mathpar}
\inferrule*[left=bSync]{
   \configm{A}{\sigma}{\mu} \trans{\una{\phi}} \configm{\skipc}{\tau}{\nu}\\
   \configm{A}{\sigma'}{\mu'} \trans{\una{\phi}} \configm{\skipc}{\tau'}{\nu'}
}{
   \configr{\syncbi{A}}{\sigma}{\sigma'}{\mu}{\mu'}
   \biTrans{}{\phi} 
   \configr{\syncbi{\skipc}}{\tau}{\tau'}{\nu}{\nu'}
}

\and
\inferrule*[left=bSyncX]{
   \configm{A}{\sigma}{\mu} \trans{\una{\phi}} \AFPFault 
\quad\mbox{or}
\quad
   \configm{A}{\sigma'}{\mu'} \trans{\una{\phi}} \AFPFault 
}{
   \configr{\syncbi{A}}{\sigma}{\sigma'}{\mu}{\mu'}
   \biTrans{}{\phi} 
   \AFPFault
}

\and
  \inferrule*[left=bCall]{ 
    (\tau|\tau') \in \phi(m)(\sigma|\sigma')
  }{ 
    \configr{\syncbi{m()}}{\sigma}{\sigma'}{\mu}{\mu'} 
      \biTrans{}{\phi} \configr{\syncbi{\skipc}}{\tau}{\tau'}{\mu}{\mu'} 
   }

\and
 \inferrule*[left=bCallX]{ 
     \PFault \in \phi(m)(\sigma|\sigma')
  }{ 
    \configr{\syncbi{m()}}{\sigma}{\sigma'}{\mu}{\mu'} 
      \biTrans{}{\phi} \PFault
   }

\and
  \inferrule*[left=bCallE]{ 
    \mbox{$\mu(m) = B \qquad \mu'(m) = B'$} 
  }{ 
    \configr{\syncbi{m()}}{\sigma}{\sigma'}{\mu}{\mu'} 
      \biTrans{}{\phi} \configr{\splitbi{B}{B'}}{\sigma}{\sigma'}{\mu}{\mu'} 
   }

\and
\inferrule*[left=bSplitL]{
   \configm{C}{\sigma}{\mu} \trans{\una{\phi}} \configm{D}{\tau}{\nu}
}{
   \configr{\splitbi{C}{C'}}{\sigma}{\sigma'}{\mu}{\mu'}
   \biTrans{}{\phi} 
   \configr{\splitbi{D}{C'}}{\tau}{\sigma'}{\nu}{\mu'}
}

\and
\inferrule*[left=bSplitR]{
   \configm{C'}{\sigma'}{\mu'} \trans{\una{\phi}} \configm{D'}{\tau'}{\nu'} 
}{
   \configr{\splitbi{\skipc}{C'}}{\sigma}{\sigma'}{\mu}{\mu'}
   \biTrans{}{\phi} 
   \configr{\splitbi{\skipc}{D'}}{\sigma}{\tau'}{\mu}{\nu'}
}

\and
\inferrule*[left=bSplitLX]{
   \configm{C}{\sigma}{\mu} \trans{\una{\phi}} \AFPFault
}{
   \configr{\splitbi{C}{C'}}{\sigma}{\sigma'}{\mu}{\mu'}
   \biTrans{}{\phi}
   \AFPFault
}

\and
\inferrule*[left=bSplitRX]{
   \configm{C'}{\sigma'}{\mu'} \trans{\una{\phi}} \AFPFault
}{
   \configr{\splitbi{\skipc}{C'}}{\sigma}{\sigma'}{\mu}{\mu'}
   \biTrans{}{\phi} 
   \AFPFault
}

\and
\inferrule*[left=bLet]{
\nu = \extend{\mu}{m}{C}\\ 
\nu' = \extend{\mu'}{m}{C'} 
}{
  \configr{ \letcombi{m}{\splitbi{C}{C'}}{DD} }{\sigma}{\sigma'}{\mu}{\mu'}
  \biTrans{}{\phi} 
   \configr{DD;\syncbi{\Endlet(m)}}{\sigma}{\sigma'}{\nu}{\nu'}
}

\and
\inferrule*[left=bIfTT]{
\sigma(E)= \True = \sigma'(E')
}{
   \configr{\ifcbi{E|E'}{CC}{DD}}{\sigma}{\sigma'}{\mu}{\mu'}
   \biTrans{}{\phi} 
   \configr{CC}{\sigma}{\sigma'}{\mu}{\mu'}
}

\and
\inferrule*[left=bIfX]{
   \sigma(E)\neq \sigma'(E') 
}{
   \configr{\ifcbi{E|E'}{CC}{DD}}{\sigma}{\sigma'}{\mu}{\mu'}
   \biTrans{}{\phi} 
   \AFPFault
}

\and
\inferrule*[left=bIfFF]{
\sigma(E)= \False = \sigma'(E')
}{
   \configr{\ifcbi{E|E'}{CC}{DD}}{\sigma}{\sigma'}{\mu}{\mu'}
   \biTrans{}{\phi} 
   \configr{DD}{\sigma}{\sigma'}{\mu}{\mu'}
}


\and
\inferrule*[left=bSeq]{
  \configr{ BB }{\sigma}{\sigma'}{\mu}{\mu'} 
   \biTrans{}{\phi}
  \configr{ CC }{\tau}{\tau'}{\nu}{\nu'} 
 }{
  \configr{ BB ; DD }{\sigma}{\sigma'}{\mu}{\mu'} 
   \biTrans{}{\phi}
  \configr{ CC ; DD}{\tau}{\tau'}{\nu}{\nu'} 
}

\and
\inferrule*[left=bSeqX]{
  \configr{ BB }{\sigma}{\sigma'}{\mu}{\mu'} 
   \biTrans{}{\phi}  \AFPFault
}{
  \configr{ BB ; DD }{\sigma}{\sigma'}{\mu}{\mu'} 
   \biTrans{}{\phi}  \AFPFault
}
\end{mathpar}
\end{small} 
\vspace*{-3ex}
\caption{Transition rules for biprograms, except biprogram loop (for which see Fig.~\ref{fig:biprogTransU}).}
\label{fig:biprogTrans}
\end{figure}

We define suitable interpretations $\phi$ for relational specifications, used to define the transition relation $\biTrans{}{\phi}$ for biprograms, analogous to the transition relation for unary programs.
If $\Phi(m)$ is a relational specification, $\phi(m)$ maps initial state-pairs to 
non-empty sets of final pairs, or fault, in accord with the specification.  The effect conditions are essentially lifted from the corresponding unary ones.
 

\begin{tdisplay}{Interpretation of relational specification}
An \dt{interpretation} $\theta$ for $\rflowty{\R}{\S}{\effe|\effe'}$, in context $\Gamma|\Gamma'$,
is a function 
\[ \theta \ : \ 
\means{\Gamma}\times\means{\Gamma'} \to \powerset^{ne} ((\means{\Gamma}\times\means{\Gamma'}) \union \{\Fault\}) 
\]
such that
if $(\tau,\tau')\in \theta(\sigma,\sigma')$ then $\sigma\successorTo\tau$ and $\sigma'\successorTo\tau'$.
Moreover 
\begin{list}{}{}
\item[(a)] $\Fault \in \theta(\sigma,\sigma')$ iff
     $\neg\some{\pi}{ \sigma|\sigma' \models_\pi  \R }$, 
   and also $\Fault \in \theta(\sigma,\sigma')$ implies $\theta(\sigma,\sigma') = \{\Fault\}$.
\item[(b)] for all $\sigma,\sigma'$ and $(\tau,\tau')$ in $\theta(\sigma,\sigma')$, 
   \begin{list}{}{}
   \item[(post)] $\all{\pi}{ (\sigma|\sigma' \models_\pi \R) \imp (\tau|\tau'\models_\pi \S )} $

   \item[(write)] $\sigma\allowTo\tau\models \effe$ and $\sigma'\allowTo\tau'\models \effe'$

   \item[(read)] 
For all $\pi,\dot{\sigma},\dot{\pi},\dot{\tau}$, \\ 
        (i) if $\sigma|\sigma'\models_\pi \R$ 
        and $\dot{\sigma}|\sigma'\models_{\dot{\pi};\pi} \R$ 
        and $(\dot{\tau},\tau')\in\theta(\dot{\sigma},\sigma')$
        then $\dot{\sigma},\sigma\allowDep\dot{\tau},\tau\models\effe$\\
        (ii) 
        if $\sigma|\sigma'\models_\pi \R$ 
        and $\sigma|\dot{\sigma}\models_{\pi;\dot{\pi}} \R$ 
        and $(\tau,\dot{\tau})\in\theta(\sigma,\dot{\sigma})$
        then $\sigma'\!,\dot{\sigma}\allowDep\tau'\!,\dot{\tau}\models\effe'$
   \end{list}
\end{list}
\vspace*{-2ex}
\end{tdisplay}
As in the case of unary interpretations, the second part of (a) is dubbed \dt{fault determinacy}.

Note that the read and write effect conditions amount to their unary counterparts, imposed on both the left and right sides.

Note also that (read)(i) is equivalent to:
For all $\dot{\sigma},\dot{\tau}$ with $(\dot{\tau},\tau')\in\theta(\dot{\sigma},\sigma')$,
if there are $\pi,\dot{\pi}$ such that $\sigma|\sigma'\models_\pi \R$ 
   and $\dot{\sigma}|\sigma'\models_{\dot{\pi};\pi} \R$ 
   then $\dot{\sigma},\sigma\allowDep\dot{\tau},\tau\models\effe$.
\emph{Mutatis mutandis} for (ii).


\begin{figure}[t!]
\begin{small}
\begin{mathpar}
\inferrule*[left=bWhL]{
  \sigma(E)=\True \\
  \sigma|\sigma'\models\P
}{
   \configr{CC}{\sigma}{\sigma'}{\mu}{\mu'}
   \biTrans{}{\phi} 
   \configr{\splitbi{\Left{BB}}{\skipc};CC}{\sigma}{\sigma'}{\mu}{\mu'}
}

\and
\inferrule*[left=bWhR]{
  \sigma|\sigma'\not\models\P \\
  \sigma'(E')=\True \\
  \sigma|\sigma'\models\P'
}{
   \configr{CC}{\sigma}{\sigma'}{\mu}{\mu'}
   \biTrans{}{\phi} 
   \configr{\splitbi{\skipc}{\Right{BB}};CC}{\sigma}{\sigma'}{\mu}{\mu'}
}

\and
\inferrule*[left=bWhTT]{
  \sigma|\sigma'\not\models\P \\
  \sigma|\sigma'\not\models\P'\\ 
  \sigma(E)= \True = \sigma'(E')
}{
   \configr{CC}{\sigma}{\sigma'}{\mu}{\mu'}
   \biTrans{}{\phi} 
   \configr{BB;CC}{\sigma}{\sigma'}{\mu}{\mu'}
}

\and
\inferrule*[left=bWhFF]{
  \sigma(E) =  \False = \sigma'(E') 
}{
   \configr{CC}{\sigma}{\sigma'}{\mu}{\mu'}
   \biTrans{}{\phi} 
   \configr{\syncbi{\skipc}}{\sigma}{\sigma'}{\mu}{\mu'}
}

\and
\inferrule*[left=bWhX]{
  \mbox{($\sigma(E)=\True$ and $\sigma'(E')=\False$ and   $\sigma|\sigma'\not\models\P$)}  
\\
  \mbox{or ($\sigma(E)=\False$ and $\sigma'(E')=\True$ and $\sigma|\sigma'\not\models\P'$)} 
}{
   \configr{CC}{\sigma}{\sigma'}{\mu}{\mu'}
   \biTrans{}{\phi} 
   \AFault
}
\end{mathpar}
\end{small}
\vspace*{-3ex}
\caption{Transition rules for loops,  in which we 
abbreviate $CC\;\equiv\; \whilecbiA{E|E'}{\P|\P'}{BB}$.
}
\label{fig:biprogTransU}
\end{figure}

Say $\phi$ is a \dt{$\Phi$-interpretation} if for each $m$ in $\dom\Phi$
with $\Phi(m)$ relational, $\phi(m)$ is an interpretation in the above sense.
In case $\Phi(m)$ is unary, $\phi(m)$ is a context interpretation in the sense defined in Sec.~\ref{app:unary}.
The biprogram transition rules are in Figs.~\ref{fig:biprogTrans} and \ref{fig:biprogTransU},
on pages \pageref{fig:biprogTrans} and \pageref{fig:biprogTransU}.
Some depend on unary transitions, for which purpose we write $\una{\phi}$ 
for the restriction of $\phi$ to unary interpretations.


\bigskip

[THIS SPACE INTENTIONALLY BLANK]

\bigskip

\clearpage

The \rn{bIf*} rules align the biprogram conditional; it faults if the same branch is not taken.
This embodies the purpose of the conditional biprogram, which is to indicate 
that the guards can be proved to agree.  Similarly for the loop transitions (Fig.~\ref{fig:biprogTransU}).
Notice that the agreement checked by conditional/loop biprograms is agreement on boolean values.  (Equality of reference values would not make sense, and agreement modulo a refperm cannot be defined because there are no refperms in the biprogram semantics.)

For a given configuration, exactly one rule is applicable.  
For context call this fact relies on two features of the semantics.  One is that a hypothesis  context maps a procedure name to a single specification, either unary or relational.
The other feature is ``fault determinacy'' of interpretations, i.e., the second part of condition (a) in the definition of interpretation for relational specifications, together with the similar condition (a) in the definition of context interpretation for unary specifications.

In all cases except where the active biprogram involves $\NEW$ or a procedure call, there is a unique outcome.  
In case of context call, $\phi$ has either a unary or a relational interpretation, and in either case the result may be nondeterministic; but it is determined modulo renaming of references, owing to the read condition in the definitions of unary/relational context interpretation.
In the case of $\NEW$, the transition rule will be one of \rn{bSync}, \rn{bSplitL}, and \rn{bSplitR}; 
any nondeterminacy in the the choice function for references
is reflected in the transitions.

\section{Relational proof rules}\label{app:relProofRules}

\begin{tdisplay}{Valid relational judgment
\hfill 
$\Phi\rHPflowtr{}{}{\P}{CC}{\Q}{\eff|\eff'}$
}  
The judgment 
is \dt{valid} iff
the following holds for all 
states $\sigma$ and $\sigma'$,
$\Phi$-interpretations $\phi$,
and refperms $\pi$ 
with  $\sigma|\sigma'\models_\pi \P$:
\begin{list}{}{}
\item[\quad(Safety)] It is not the case that 
$\configr{CC}{\sigma}{\sigma'}{\_\,}{\,\_} \biTranStar{}{\phi} \,\Fault$.
\end{list}
And for 
all $\tau,\tau'$ such that 
$\configr{CC}{\sigma}{\sigma'}{\_}{\_} \biTranStar{}{\phi} 
\configr{\syncbi{\skipc}}{\tau}{\tau'}{\_}{\_}$  
\begin{list}{}{}
\item[\quad(Post)] $\tau|\tau' \models_\pi \Q$  
\item[\quad(Write Effect)]
$\sigma\allowTo\tau\models \eff$ and $\sigma'\allowTo\tau'\models \eff'$ 

\item[\quad (Read Effect)]
For any $\rho,\dot{\sigma},\dot{\tau}$, 
\\ 
(i) if
$\configr{CC}{\dot{\sigma}}{\sigma'}{\_}{\_} \biTranStar{}{\phi} 
\configr{\syncbi{\skipc}}{\dot{\tau}}{\tau'}{\_}{\_}$ 
and 
$\dot{\sigma}|\sigma'\models_{(\rho;\pi)} \P$
then
$\dot{\sigma},\sigma\allowDep\dot{\tau},\tau\models\eff$
\\
(ii) if
$\configr{CC}{\sigma}{\dot{\sigma}}{\_}{\_} \biTranStar{}{\phi} 
\configr{\syncbi{\skipc}}{\tau}{\dot{\tau}}{\_}{\_}$ 
and 
$\sigma|\dot{\sigma}\models_{(\pi;\rho)} \P$
then
$\sigma'\!,\dot{\sigma}\allowDep\tau'\!,\dot{\tau}\models\eff'$
\end{list}
\vspace*{-2ex}
\end{tdisplay}

In addition to the relational proof rules in Fig.~\ref{fig:proofrulesR}, 
we give in Fig.~\ref{fig:proofrulesRX} some additional rules that have been proved sound.

\begin{figure}
\begin{mathpar}

\inferrule*[left=rCall]{}{
\quad \rflowtr{\P}{m}{\Q}{\eff} \rHPflowtr{}{}{\P}{\syncbi{m()}}{\Q}{\eff} 
}

\and
\inferrule*[left=rLater]
{\Phi\rHPflowtr{}{}{\P}{CC}{\Q}{\eff|\eff'} 
}
{ \Phi\rHPflowtr{}{}{\later\P}{CC}{\later\Q}{\eff|\eff'} 
}

\and
\inferrule*[left=rFrame]
{ \Phi \rHPflowtr{}{}{\P}{CC}{\Q}{\eff|\eff'} \\
  \P\models \fra{\effe|\effe'}{\R} \\
  \P\land\R \imp \leftF{(\ind{\effe}{\eff})} \land \rightF{(\ind{\effe'}{\eff'})} 
}{
\Phi \rHPflowtr{}{}{\P\land \R}{CC}{\Q\land \R}{\eff|\eff'}
}

\and
\inferrule*[left=rEmb]{ 
   \Phi\HPflowtr{}{}{P}{C}{Q}{\eff}\\
   \Phi\HPflowtr{}{}{P'}{C'}{Q'}{\eff'}\\
}{ \Phi\rHPflowtr{}{}{\leftF{P}\land\rightF{P'}}{\splitbi{C}{C'}}{\leftF{Q}\land\rightF{Q'}}{\eff|\eff'}
}

\and
\inferrule*[left=rSeq]
{\Phi\rHPflowtr{}{}{\P}{CC_1}{\P_1}{\eff_1|\eff'_1} \\ 
\Phi\rHPflowtr{}{}{\P_1}{CC_2}{\Q}{\eff_2|\eff'_2} \\
\eff_2 \mbox{ is } \Left{\P}/\eff_1\mbox{-immune} \\
\eff'_2 \mbox{ is } \Right{P}/\eff'_1\mbox{-immune} 
}
{ \Phi\rHPflowtr{}{}{\P}{\seqc{CC_1}{CC_2}}{\Q}{\eff_1, \eff_2|\eff'_1,\eff'_2}
}
\and 
\inferrule*[left=rWh]
{
\Phi \rHPflowtr{}{}{\Q\land\P\land\leftF{E}}{\splitbi{\Left{CC}}{\skipc}}{\Q}{\eff|\,} \\
\Phi \rHPflowtr{}{}{\Q\land\P'\land\rightF{E'}}{\splitbi{\skipc}{\Right{CC}}}{\Q}{\,|\eff'} \\
\Phi \rHPflowtr{}{}{\Q\land\neg\P\land\neg\P'\land\leftF{E}\land\rightF{E'}}{CC}{\Q}{\eff|\eff'} \\
\Q\imp E\eqbi E' \lorbi (\P\land\leftF{E}) \lorbi (\P'\land\rightF{E'}) \\
\Q\proves\fra{\eff|\eff'}{\P}\\
\eff\mbox{ is }\Left{\P}/\eff\mbox{-immune}\\
\later\P\imp\P \\
\Q\proves\fra{\eff|\eff'}{\P'}\\
\eff'\mbox{ is }\Right{\P'}/\eff'\mbox{-immune} \\
\later\P'\imp\P' 
}{
\Phi \rHPflowtr{}{}{\Q}{\whilecbiA{E\smallSplitSym E'}{\P\smallSplitSym \P'}{CC}}{\Q}{\eff,\ftpt(E)|\eff',\ftpt(E')}
}

\inferrule*[left=rEqu]
{\Phi\rHPflowtr{}{}{\P}{\splitbi{C}{C'}}{\Q}{\eff|\eff'} 
\\ C\uequiv D 
\\ C'\uequiv D' 
}
{ \Phi\rHPflowtr{}{}{\P}{\splitbi{D}{D'}}{\Q}{\eff|\eff'} 
}

\and
\inferrule*[left=rConseq]{ 
 \Phi\rHPflowtr{}{}{\P}{CC}{\Q}{\eff|\eff'}\\
 \R\imp \P \\
 \Q \imp \S \\
 \P\models (\eff|\eff')\leq(\effe|\effe') \\
}{ 
 \Phi\rHPflowtr{}{}{\R}{CC}{\S}{\effe|\effe'}
}

\end{mathpar}
\caption{Additional relational proof rules.}\label{fig:proofrulesRX}
\end{figure}

The \rn{rLater} rule is used (with \rn{rConseq} and  $\later\later\P\imp\later\P$) to derive a variation on \rn{rSeq} 
where the intermediate relation has $\later$ (see footnote~\ref{fn:later}).

Rule \rn{rConseq} includes a subeffect judgment 
$ \P\models (\eff|\eff')\leq(\effe|\effe')$ which is a direct generalization 
of subeffects in the unary logic.

Rule \rn{rEqu} uses unconditional program equivalence to rewrite the commands in a split,
if they differ only in the way their control flow is expressed, i.e., their behavior in all contexts is the same.  
Commands $C,C'$ are \dt{unconditionally equivalent},
written
$C \uequiv C'$, 
iff for all $\sigma,\tau,\phi,D$ we have
\begin{list}{}{}
\item[(a)]
If $\configm{C}{\sigma}{\mu} \tranStar{\phi} \configm{\skipc}{\tau}{\mu}$ then
$\configm{C'}{\sigma}{\mu} \tranStar{\phi} \configm{\skipc}{\tau}{\mu}$.
\item[(b)]
If $\configm{C}{\sigma}{\mu} \tranStar{\phi} \configm{D}{\tau}{\nu} \trans{\phi} \Fault$
then $\configm{C'}{\sigma}{\mu} \tranStar{\phi} \configm{D'}{\tau}{\nu} \trans{\phi} \Fault$
for some $D'$.
\item[(c)] \emph{Mutatis mutandis} for $C'/C$.
\end{list}

\begin{lemma}
\label{lem:uncond}
$\uequiv$ is an equivalence relation and \\
\mbox{(a) } \( \WHILE\; E\;\DO\;C\;\OD
\uequiv
\WHILE\; E \;\DO\; C; \WHILE\; E \land E0 \;\DO\; C \;\OD\; \OD
\)
\\
\mbox{(b) } \(\WHILE\; E \;\DO\; C \;\OD
\uequiv
\IF\; E\;\DO\; C\;\FI; \WHILE\; E\;\DO\; C \;\OD
\)
\end{lemma}
We also have $C \uequiv \skipc; C; \skipc$.
This is an instance of reflexivity, because we identify $\skipc; C; \skipc$ with $C$.

The weaving relation $\weave$ is defined inductively by axioms and congruence rules.
Here is the complete list of weaving axioms, each of which replaces a split by one of the other biprogram forms:
\[\begin{array}{l}
\splitbi{A}{A} \weave \syncbi{A} \\[.5ex]
\Splitbi{C;D}{C';D'} \weave \splitbi{C}{C'};\splitbi{D}{D'} \\[.5ex]
\Splitbi{ \ifc{E}{C}{D} }{ \ifc{E'}{C'}{D'} } 
   \weave \ifcbi{E\smallSplitSym E'}{\splitbi{C}{C'} }{ \splitbi{D}{D'} }
\\[.5ex]
\Splitbi{ \whilec{E}{C} }{ \whilec{E'}{C'} }
   \weave \whilecbiA{E\smallSplitSym E'}{\P\smallSplitSym \P'}{\splitbi{C}{C'}}
\\[.5ex]
   \Splitbi{ \letcom{m}{B}{C} }{ \letcom{m}{B'}{C'} }
    \weave \letcombi{m}{\splitbi{B}{B'}}{ \splitbi{C}{C'} } 
 \end{array}\]
Here is the complete list of congruence rules.  
Each is formulated in terms of a single sub-biprogram, for technical convenience.
The premise in each case is $BB\weave CC$;
the conclusions are
\[\begin{array}{l}
BB;DD \weave CC;DD\\
DD;BB \weave DD;CC\\
\ifcbi{E\smallSplitSym E'}{BB}{DD} \weave \ifcbi{E\smallSplitSym E'}{CC}{DD} \\
\ifcbi{E\smallSplitSym E'}{DD}{BB} \weave \ifcbi{E\smallSplitSym E'}{DD}{CC} \\
\whilecbiA{E\smallSplitSym E'}{\P\smallSplitSym \P'}{BB} \weave
  \whilecbiA{E\smallSplitSym E'}{\P\smallSplitSym \P'}{CC} \\
\letcombi{m}{\splitbi{B}{B'}}{ BB } \weave \letcombi{m}{\splitbi{B}{B'}}{ CC }
\end{array}\]

At this point we have defined the syntax and semantics of the logic.
In the following sections we apply the logic to examples.
Then Sections~\ref{app:consistency} and~\ref{app:sound} justify the semantics and prove the rules sound.

\section{Proof for example of ``dissonant'' loop}\label{app:dissonant}

Example $C_1$, discussed in Secs.~\ref{sec:unaryLogic} and~\ref{sec:overview},
shows the use of alignment guards to achieve intermittent alignment of iterations.
It also uses a region variable, $r$, to express frame conditions for heap locations.
\labf{C_1}{
\begin{array}[t]{l}
s:=0;\WHILE\ p\not=\NULL\ \DO\\
\quad\quad\quad\quad\IF\ \neg p.del\ \THEN\ s:=s+p.val\ \FI \\
\quad\quad\quad\quad p:=p.nxt;\\
\quad\quad\quad \OD
\end{array}}
Let
\labf{\eff}{\rw{s},\rw{p},\rd{r},\rd{r\Img val},\rd{r\Img nxt},\rd{r\Img del}}
We want to prove the judgment 
\[\begin{array}{l}
\splitbi{C_1}{C_1} \: : \: \Both(p\in r\land r\Img nxt\subseteq r)\land 
listnd(p)\eqbi listnd(p)
\rspecSym s\eqbi s ~ [\eff]
\end{array}
\]
Let 
\labf{B}{\IF\ \neg p.del\ \THEN\ s:=s+p.val\ \IF}
\labf{D}{B; p:=p.nxt;}
\labf{DD}{\begin{array}[t]{l}
\WHILE\ p\not=\NULL\mid p\not=\NULL\ \cdot\ \leftF{p.del}\mid \rightF{p.del}\ \DO\\
\quad\quad\quad\quad\quad\quad\splitbi{D}{D}\ \OD
\end{array}}
We have $\splitbi{C_1}{C_1}\weave^* \syncbi{s:=0};DD$. By rule \rn{rWeave} (in
Fig.~\ref{fig:proofrulesR}, Sec.~\ref{sec:relLogic}), it is enough to prove 
\[\begin{array}{l}
\syncbi{s:=0};DD \: : 
\quad \Both(p\in r\land r\Img nxt\subseteq r)\land 
listnd(p)\eqbi listnd(p)
\rspecSym s\eqbi s ~ [\eff]
\end{array}
\]
By a unary judgment 
then \rn{rEmb} (Fig.~\ref{fig:proofrulesR})
and \rn{rConseq} (Sec.~\ref{app:relProofRules}) we have
\[\syncbi{s:=0}\: : \: \True\rspecSym s\eqbi s  ~ [\wri{s}] \]
By \rn{rFrame} and \rn{rConseq} (for subeffects)
we have 
\[\begin{array}{l}
\syncbi{s:=0} \: : \: \Both(p\in r\land r\Img nxt\subseteq r)\land listnd(p)\eqbi listnd(p)\rspecSym\\
\Both(p\in r\land r\Img nxt\subseteq r)\land listnd(p)\eqbi listnd(p) \land s\eqbi s ~ [\eff\mid\eff]
\end{array}
\]
By \rn{rSeq}, it is enough to show that 
\[DD\: : \: Q\rspecSym Q ~ [\eff\mid\eff]\]
where
\labf{Q}{\Both(p\in r\land r\Img nxt\subseteq r)\land listnd(p)\eqbi listnd(p) \land s\eqbi s}
which will serve as loop invariant in \rn{rWh}.
We have 
\[Q\imp(p=\NULL\eqbi p=\NULL)\lor\leftF{(p.del\land p\not=\NULL)}\lor\rightF{(p.del\land p\not=\NULL)}\]
The consequent follows from $listnd(p)\eqbi listnd(p)$.
It remains to prove the following three judgments:
\begin{eqnarray}
\label{leftWh}
&&\splitbi{D}{\skipc}\: : \: Q\land\leftF{(p.del\land p\not=\NULL)}\rspecSym Q ~ [\eff\mid]\\
\label{rightWh}
&&\splitbi{\skipc}{D}\: : \: Q\land\rightF{(p.del\land p\not=\NULL)}\rspecSym Q ~ [\mid\eff]\\
\label{bothWh}
&&\splitbi{D}{D}\: : \: Q\land\Both{(\neg p.del\land p\not=\NULL)}\rspecSym Q ~ [\eff\mid\eff]
\end{eqnarray}
(That is, the three premises of \rn{rWh}, after some simplification.)
To prove these, we use unary judgments and then embed.
Add variables $t$ and $l$ for use in assertions.
By \rn{If} we have 
\[B\: : \: p.del\land p\neq \NULL\land t=s\ \leadsto t=s ~ [\wri{s}]\]
By \rn{Frame} we get 
\[\begin{array}{l}
B\: : \: 
p.del\land p\neq \NULL \land t=s 
 \land  
(p\in r\land r\Img nxt\subseteq r\land l=listnd(p)) 
\\
\leadsto
t=s \land (p\in r\land r\Img nxt\subseteq r \land l=listnd(p)) ~ [\effe]
\end{array}
\]
where $\effe$ is $\eff\setminus\wri{p}$. 
For $p:=p.nxt$ we use the local axiom for field read, together with \rn{Conseq} to reason 
using the mathematical fact
$p\neq \NULL \land p.del \imp listnd(p)=listnd(p.nxt)$.
Then by \rn{Seq} we get
\[\begin{array}{l}
D\: : \: p\in r\land r\Img nxt\subseteq r\land l=listnd(p) \land t=s \land p.del\land\\ p\neq \NULL \leadsto
p\in r\land r\Img nxt\subseteq r \land l=listnd(p) \land t=s ~ [\eff]
\end{array}
\]
On the other hand, we have
\[\begin{array}{l}
\skipc\: : \: p\in r\land r\Img nxt\subseteq r\land l=listnd(p)\land t=s\\\leadsto
p\in r\land r\Img nxt\subseteq r \land l=listnd(p)\land t=s ~ [\eff]
\end{array}
\]
Since $t$ and $l$ are not written, we can use \rn{rEmb}, \rn{rFrame}, and \rn{rConseq} to get 
(\ref{leftWh}) and (\ref{rightWh}).
To be precise, we need the usual $\exists$-elimination rule (e.g.,~\cite{RegLogJrnI}),
to eliminate the preconditions $t=s$ and $l=listnd(p)$.
For (\ref{bothWh}) one could weave $\splitbi{D}{D}$ but there's no need to.  Similar steps to the preceding ones
can be used to obtain (\ref{bothWh}).

Finally, the sequence and loop rules have immunity side conditions,
variations on $\eff$ being $true/\eff$-immune.  
In brief, this condition simplifies to $true$ because the assigned variables $s$ and $p$ do not occur in region expressions in $\eff$.
In more detail, 
notice that the only image expressions 
in $\eff$ are $r\Img val,r\Img nxt,r\Img del$. 
Applying the $\indSymbol$ operator, we find that
$\ind{\rd{r\Img val, r\Img nxt, r\Img del, r}}{\wri{s,p}}$ 
trivially simplifies to $true$. 

\section{Loop tiling example}

Loop tiling is a compiler optimization that changes program structure.
Here is an example from \cite{BartheCK13}.
\[
\begin{array}{l}
C_2\eqdef x := 0;\WHILE\; x < N*M\;\DO\;  a[x] := f(x); x\increment\;\OD\\
C_2'\eqdef i:=0;\WHILE\; i < N\;\DO\;j := 0;
     \; \WHILE\;j < M\;\DO\; A[i, j] := f(i*M + j); j\increment\;\OD;
     \; i\increment\;\OD
\end{array}
\]
These are not equivalent, but are equivalent modulo change of data representation.
We express this by the judgment 
\begin{equation}
\label{eq:tile}
\splitbi{C_2}{C_2'}: \: \Both{true}  \rspecSym \R(M * N, N, M)
\end{equation}
\labf{\R(x, i, j)}{
\!\!\!
\begin{array}[t]{l}
\all{l,r,s}{0\leq l < x \land 0\leq r < i \land 0\leq s < i \land 
 l=r *  M + s \imp a[l] \eqbi A[r,s]}
\end{array}
}
To prove it, we rely on some unconditional program equivalences
that change the control structure without altering the order of atomic commands
(see Lemma~\ref{lem:uncond}).
First, rewrite $C_2$ and $C_2'$ to 
\labf{C_3}{
\begin{array}[t]{l@{\,}l}
x := 0; 
&
\WHILE\; x < N*M\;\DO \\
&\quad\skipc; a[x] := f(x); x\increment;\\
&\quad\WHILE\; x < M*N \land x \% M \neq 0\;\DO \\
&\quad\quad a[x] := f(x); x\increment\;\OD;\\
&\quad\skipc;\OD
\end{array}
}
\labf{C_3'}{
\begin{array}[t]{l@{\,}l}
i := 0; &
\WHILE\; i < N\; \DO\; j := 0; \\
&\quad\IF\; j < M\;\THEN\; A[i, j] := f(i*M + j); j\increment\;\FI\\
&\quad\WHILE\; j < M\;\DO\; A[i, j] := f(i*M + j); j\increment\;\OD\\
&\quad i\increment\;\OD
\end{array}
}
These rewrites change the control state without altering the trace of data states (modulo stuttering).
Formally, $C_2\uequiv C_3$ and $C_2'\uequiv C_3'$, where $\uequiv$ means unconditional equivalence
and is defined in Sec.~\ref{app:relProofRules}.
We apply the relational proof rule \rn{rEqu}, that is,
$\splitbi{C_2}{C_2'}$ satisfies the 
specification in (\ref{eq:tile}) if 
$\splitbi{C_3}{C_3'}$ does.
The rewrites are chosen so that we can weave $\splitbi{C_3}{C_3'}$  to 
the carefully aligned biprogram
\[ 
\begin{array}[t]{l} 
\Splitbi{x := 0}{i := 0};\\
\WHILE\; x < N * M \mid i < N \; \DO\\
\quad\Splitbi{\skipc}{j := 0};\\
\quad \big(a[x] := f(x);\;x\increment 
\quad \big|
\begin{array}[t]{l}
\IF\;j < M\;\THEN\;  
A[i, j] := f(i *  M + j);\;j\increment ~ \FI
\end{array}
\big)\\
\quad\WHILE\; x > M * N\;\land x \% M \neq 0 \mid j < M \; \DO\\
\quad\quad\Splitbi{a[x] := f(x); x\increment}{A[i, j] := f(i * M + j); j\increment} \\
\quad\OD;\\
\quad\Splitbi{\skipc}{i\increment}\\
\OD
\end{array}
\] 
To prove that this satisfies $\Both{true}  \rspecSym \R(M * N, N, M)$, 
we use 
\(
x\eqbi i\times M + j\land \R(x,i,j)
\) 
as the invariant for inner loop and 
\(
x\eqbi i\times M \land \R(x,i,0)
\)
as the invariant of the outer loop.



\section{Proof for Stack example}\label{app:stack}

In this section we provide a more detailed proof of equivalence for data representation example.
The sketch of the proof given in Section~\ref{sec:overview} glossed  
over dynamic allocation.
To fully consider dynamic allocation we use rule \rn{Alloc} in Fig.~\ref{fig:unaryRules} for unary judgments and the axiom of allocation 
mentioned in Sec.~\ref{sec:relLogic}.
According to these rules we need to change $\effe$ in (\ref{eq:push}) to the following:
\[\effe\eqdef  \rw{rep,size,rep\Img\allfields,\lloc} \] 
and we also add $\later$ to the postcondition in (\ref{eq:ggx}).
We will use the fact, noted in Sec.~\ref{sec:relForm}, that 
\begin{equation}\label{eq:app:mono-distrib}
\later\P\land\Q \imp \later(\P\land\Q) \qquad \mbox{ is valid (for any $\P$ and any monotonic $\Q$)}
\end{equation}

Recall that $\Phi$ on page \pageref{disp:Phi:pg}
gives the relational specification for $push$. 
To show $\splitbi{B}{B'}$ satisfies the relational specification for $push$,
we weave it to 
\[ 
\begin{array}{c}
\big( \begin{array}{l}
            top := \NEW\ Node(top,x); \\
            rep := rep \unionC \{top\}
          \end{array}
\big |
          \begin{array}{l}
             \IF\ slots = \NULL\ \THEN\ \ldots \FI;\\
             slots[free\increment] := x
          \end{array}
\big) ;
\\
\syncbi{size\increment}
\end{array}
\]
Let $D$ and $D'$ name the split parts, so the woven code 
has the form $\splitbi{D}{D'};\syncbi{size\increment}$.
By unary reasoning and applying rules \rn{rEmb}, \rn{rFrame}, and \rn{rConseq}, we can show that $\splitbi{D}{D'}$ satisifies
\begin{equation}
\label{eq:ddd'}
\begin{array}{l}
   \Both{R}
   \land \L
   \land \Agr(size,x)
\rspecSym
   LtR
   \land \leftF{I}
   \land \rightF{I'_1}
   \land \Agr(size,x)
\\ {}
[\rw{top,rep,rep\Img\allfields}  
\mid \rw{slots,free,rep}]         
\end{array}
\end{equation}
Here $I'_1$ is a slight variant of $I'$, where the last conjunct is $free = size + 1$.
By unary logic we get
\[ 
\begin{array}{l}
size\increment : \:
   I 
   \leadsto
   (size=\keyw{old}(size) + 1)
   \land I 
   ~[\rw{size}]
\\ 
size\increment : \:
   I'_1 
   \leadsto
   (size=\keyw{old}(size) + 1)
   \land I' 
   ~[\rw{size}]
\end{array}
\]
An embedding rule lifts these to a relational judgment with agreements, 
and then \rn{rFrame} for $LtR$ yields
\[
\syncbi{size\increment} : \:
\begin{array}[t]{l}
   LtR 
   \land \leftF{I} 
   \land \rightF{I'_1} 
\\
   \rspecSym
   \Both{(size=\keyw{old}(size) + 1)}  \land 
   \Agr size \land \L 
 ~[\rw{size}]
\end{array}
\]
which by \rn{rConseq} shows that $\syncbi{size\increment}$ satisfies
\begin{equation}
\label{eq:size}
   LtR 
   \land \leftF{I} 
   \land \rightF{I'_1} 
   \land \Agr size
   \rspecSym
   \Both{S}  \land 
   \Agr size \land \L ~ [\rw{size}]
\end{equation}
From (\ref{eq:ddd'}) and (\ref{eq:size}) 
by rule \rn{rSeq} we get that 
$\splitbi{D}{D'};\syncbi{size\increment}$ satisfies the specification for $push$.  Hence by the weaving rule so does $\splitbi{B}{B'}$.

Now we aim to prove the revised version of (\ref{eq:ggx}), that is, 
\[
\Syncbi{Cli}\;:\; \Both{(size=0)} \land \L
\rspecSym \later(\Agr(p,size) \land \Agr \sing{p}\Img \allfields \land \L)
\]
where $\Syncbi{Cli}$ is the fully aligned biprogram 
\[ 
\syncbi{push(1)}; 
  \syncbi{p := \NEW\ Node(\NULL,2)};
  \syncbi{p.val := 3}; 
  \syncbi{push(4)} 
\]

The first command is a method call to $push$. From rule \rn{rCall} and the relational specification $\Phi$ of $push$, we derive 
\[
\syncbi{push(1)}\;:
\begin{array}[t]{l}
   \Both{R}
   \land \Agr size
   \land \L
\rspecSym
\later(
   \Both{S} 
   \land \Agr size
   \land \L)
\\ {}
[\effe, \rw{top} 
\mid 
\effe, \rw{slots,free}] 
\end{array}
\]
Notice that $\Both{(size=0)} \land \L$ implies $\Both{R} \land \Agr size \land \L$. Using rule \rn{rConseq} 
and a little sleight of hand we get 
\begin{equation}\label{eq:FPush}
\syncbi{push(1)}\;:
\begin{array}[t]{l}
  \Both{(size=0)} \land \L
\rspecSym\\
\later(   \Both(size =1 \land r=\lloc\land rep\subseteq r)
   \land \Agr size
   \land \L)
\\ {}
[\effe, \rw{top} 
\mid 
\effe, \rw{slots,free}] 
\end{array}
\end{equation}
The sleight of hand is to introduce a fresh ghost variable $r$ to snapshot $\lloc$.  (The condition $rep\subseteq r$ follows 
from $r=\lloc$.)  An entirely rigorous proof would add an assignment to $r$ but for clarity we will skip that.

The second command in $\Syncbi{Cli}$ is allocation. Using the axiom of allocation mentioned above, we derive
\[
\begin{array}{l}
\syncbi{p := \NEW\ Node(\NULL,2)}\;: \Both (r=\lloc)\rspecSym\\
\later (\Both(p\not=\NULL\land p\notin r) \land \Agr (p, \sing{p}\Img \allfields))\\
{[\wri{p},\rw{alloc}]}
\end{array}
\]
We aim to frame $rep\subseteq r$ over this judgment. Note that $\rd{rep,r}$ frames $rep\subseteq r$ and 
$\ind{\rd{rep,r}}{\wri{p},\rw{alloc}}$. By rules \rn{rFrame}, and \rn{rConseq} using the validity (\ref{eq:app:mono-distrib}), we get 
\[
\begin{array}{l}
\syncbi{p := \NEW\ Node(\NULL,2)}\;: \Both(r=\lloc\land rep\subseteq r)\rspecSym\\
\later (\Both(p\not=\NULL\land p\notin r\land rep\subseteq r) \land \Agr (p, \sing{p}\Img \allfields))\\
{[\wri{p},\rw{alloc}]}
\end{array}
\]
Using \rn{rConseq} we rewrite the postcondition to get 
\begin{equation}
\label{eq:allocp}
\begin{array}{l}
\syncbi{p := \NEW\ Node(\NULL,2)}\;: \Both(r=\lloc\land rep\subseteq r)\rspecSym\\
\later (\Both(p\not=\NULL\land p\notin rep) \land \Agr (p, \sing{p}\Img \allfields))\\
{[\wri{p},\rw{alloc}]}
\end{array}
\end{equation}
For the third command we use the unary \rn{FieldUpd} and \rn{rEmb} and \rn{rConseq} to get
\[
\syncbi{p.val:=3}\;:\;\Both(p\not=\NULL)\land\Agr(p)\rspecSym \Agr(p.val)[\rd{p},\wri{p.val}]
\]
We aim to use \rn{rFrame} on this last judgment to add $\Agr p,\sing{p}\Img nxt\land \Both(p\notin rep)$ to precondition and postcondition of this last 
judgment. Note that $\rd{p,\sing{p}\Img nxt, rep}$ frames $\Agr (p,\sing{p}\Img nxt)\land \Both (p\notin rep)$.
And $\ind{\rd{p,\sing{p}\Img nxt},rep}{\rd{p},\wri{p.val}}$ simplifies to true. 
Using these two facts and rule \rn{rFrame} we derive 
\[
\begin{array}{ll}
\syncbi{p.val:=3}\;:\Both(p\not=\NULL\land p\notin rep)\land\Agr(p,\sing{p}\Img nxt)\rspecSym\\
\Both(p\notin rep)\land\Agr(p,\sing{p}\Img nxt,\sing{p}\Img val )\\
{[\rd{p},\wri{p.val}]}
\end{array}
\]
Now we use rule \rn{rConseq} and \rn{rLater} to get 
\begin{equation}\label{eq:pvalch}
\begin{array}{ll}
\syncbi{p.val:=3}\;:\\
\later(\Both(p\not=\NULL\land p\notin rep)\land\Agr(p,\sing{p}\Img nxt))\\
\rspecSym
\later(\Both(p\notin rep)\land\Agr(p,\sing{p}\Img \allfields ))
{[\rd{p},\rw{p.val}]}
\end{array}
\end{equation}
The judgments (\ref{eq:allocp}) and (\ref{eq:pvalch}) are now ready to be unified by rule \rn{rSeq}. 
So, from \rn{rSeq} we derive
\[
\begin{array}{l}
\syncbi{p := \NEW\ Node(\NULL,2)};\syncbi{p.val:=3}:
\Both(r=\lloc\land rep\subseteq r)\\
\rspecSym \later(\Both(p\notin rep) \land \Agr (p, \sing{p}\Img \allfields))[\rw{p, alloc}]
\end{array}
\]
Actually we need to use the general form of \rn{rSeq} which, like \rn{Seq}, lets us remove $\rw{p.val}$ from the overall effects of this last judgment.
The general form of the \rn{rSeq} indicates that if in the first command some references are allocated and the second commands writes some fields of these newly allocated references then the second command's effects should have writes as well as reads of these fields.

Now we frame $\Both(size=1)\land\Agr size \land \L$ over the last judgment. 
\[
\begin{array}{l}
\syncbi{p := \NEW\ Node(\NULL,2)};\syncbi{p.val:=3}:\\
\Both(size = 1 \land r=\lloc\land rep\subseteq r)\land \Agr size \land \L\\
\rspecSym\later(\Both(p\notin rep) \land \Agr (p, \sing{p}\Img \allfields))
            \land \Both(size=1) \land \Agr size \land \L \\
{ } [\rw{p, alloc}]
\end{array}
\]
Now by \rn{rConseq} we get.
\[
\begin{array}{l}
\syncbi{p := \NEW\ Node(\NULL,2)};\syncbi{p.val:=3}:\\
\Both(size = 1 \land r=\lloc\land rep\subseteq r)\land \Agr size \land \L\\
\rspecSym\later(\Both(p\notin rep) \land \Agr (p, \sing{p}\Img \allfields, size)
            \land \Both(size=1) \land \L) \\
{ } [\rw{p, alloc}]
\end{array}
\]
Now using \rn{rLater}, followed by \rn{rConseq} instantiating the valid formula
$\later\later\P\imp\later\P$ we get 
\[
\begin{array}{l}
\syncbi{p := \NEW\ Node(\NULL,2)};\syncbi{p.val:=3}:\\
\later(\Both(size = 1 \land r=\lloc\land rep\subseteq r)\land \Agr size \land \L)\\
\rspecSym\later(\Both(p\notin rep) \land \Agr (p, \sing{p}\Img \allfields, size)
            \land \Both(size=1) \land \L) \\
{ } [\rw{p, alloc}]
\end{array}
\]
Using \rn{rSeq} on (\ref{eq:FPush}) and the last judgment we get
\[
\begin{array}{l}
\syncbi{push(1)};\syncbi{p := \NEW\ Node(\NULL,2)};\syncbi{p.val:=3}:\\
  \Both{(size=0)} \land \L
\rspecSym\\
\later(\Both(size = 1 \land p\notin rep) \land \Agr (p, \sing{p}\Img \allfields, size)\land \L)\\
{[\effe, \rw{top,p} 
\mid 
\effe, \rw{slots,free,p}] 
}
\end{array}
\]
Using \rn{rConseq} to remove $\Both( p\notin rep)$ from postcondition, we get 
\begin{equation}
\label{eq:threecoms}
\begin{array}{l}
\syncbi{push(1)};\syncbi{p := \NEW\ Node(\NULL,2)};\syncbi{p.val:=3}:\\
  \Both{(size=0)} \land \L
\rspecSym\later (\Both (size = 1)\land \Agr (p, \sing{p}\Img \allfields, size)\land \L)\\
{[\effe, \rw{top,p} 
\mid 
\effe, \rw{slots,free,p}] 
}
\end{array}
\end{equation}
For the last command we use \rn{rCall} similar to (\ref{eq:FPush}) to get 
\[
\syncbi{push(4)}\;:
\begin{array}[t]{l}
   \Both{R}
   \land \Agr size 
   \land \L
\rspecSym
\later(
   \Both{S} 
   \land \Agr size
   \land \L)
\\ {}
[\effe, \rw{top} 
\mid 
\effe, \rw{slots,free}] 
\end{array}
\]
We use \rn{rConseq} we derive
\[
\syncbi{push(4)}\;:
\begin{array}[t]{l}
   \Both{(size = 1)}
   \land \Agr size 
   \land \L
\rspecSym
\later( \Agr size
   \land \L)
\\ {}
[\effe, \rw{top} 
\mid 
\effe, \rw{slots,free}] 
\end{array}
\]
Using \rn{rFrame} and \rn{rConseq} we derive
\[
\syncbi{push(4)}\;:
\begin{array}[t]{l}
   \Both{(size = 1)}
   \land \Agr (p, \sing{p}\Img \allfields, size)\land \L
\rspecSym\\
 \Agr (p, \sing{p}\Img \allfields, size)\land \L
 {}
[\effe, \rw{top} 
\mid 
\effe, \rw{slots,free}] 
\end{array}
\]
To add the last command to (\ref{eq:threecoms}), we use \rn{rSeq} to derive 
\[
\begin{array}{l}
\Syncbi{Cli}:
  \Both{(size=0)} \land \L
\rspecSym\Agr (p, \sing{p}\Img \allfields, size)\land \L\\
{[\effe, \rw{top,p} 
\mid 
\effe, \rw{slots,free,p}] 
}
\end{array}
\]
This finishes the proof.

\section{Semantic consistency theorem}\label{app:consistency}

The ultimate point of the relational logic is to prove relational properties of ordinary
programs.  Theorem~\ref{thm:biprogram-soundness} connects biprogram semantics with unary semantics,
for hypothesis contexts that have only unary specifications.  Such contexts model ambient libraries, and are meaningful for biprograms as well as for ordinary commands.  By contrast, relational hypotheses can be introduced by rule \rn{rLink} for modular relational reasoning 
about linked subprograms.   

\begin{theorem}[semantic consistency]\label{thm:biprogram-soundness}
Suppose $\Phi$ has only unary specifications.
Suppose  $\Phi\rHPflowtr{}{}{\P}{\splitbi{C}{C'}}{\Q}{\eff|\eff'}$ is valid.
Consider any $\Phi$-interpretation $\phi$.
Consider any $\sigma,\sigma',\pi$ with $\sigma|\sigma'\models_\pi\P$.
If $\configm{C}{\sigma}{\_}\tranStar{\phi}\configm{\skipc}{\tau}{\_}$ and 
$\configm{C'}{\sigma'}{\_}\tranStar{\phi}\configm{\skipc}{\tau'}{\_}$
then $\tau|\tau'\models_\pi\Q$.
Furthermore, if $C$ does not diverge from $\sigma$ 
then both of these initial configurations are safe 
(i.e., do not fault).
\end{theorem}

To prove the theorem we use lemmas that connect biprogram and unary semantics; these lemmas are also used in proving soundness for some of the proof rules.

A \dt{trace} is a consecutive sequence of configurations,
numbered from 0. 
Let $T$ be a biprogram trace and $U,V$ unary traces.  
A \dt{schedule of $U,V$ for $T$} is a pair $l,r$ with $l:(\dom(T))\to(\dom(U))$ and $r:(\dom(T))\to(\dom(V))$, each surjective and monotonic.
(For example, look at the sketches in Sec.~\ref{sec:biprogramSem}, where the 
dashed lines indicate how indices of the biprogram trace are mapped by $l$ and $r$.)
A schedule $l,r$ is an \dt{alignment} 
of $U,V$ for $T$, written $align(l,r,T, U, V)$,
iff 
$U_{l(i)} = \Left{T_i}$ and $V_{r(i)} = \Right{T_i}$
for all $i$ in $\dom(T)$.

It is convenient to classify the biprogram transition rules as follows.
Rules \rn{bSeq} and \rn{bSeqX} simply close the transitions under command sequencing.
All the other rules apply to a non-sequence biprogram of some form; for any 
biprogram configuration that is not terminated, there is a unique one of these rules that applies.
In the case of context calls, this is a consequence of a condition (fault determinacy) in the definition of context interpretation.  
We dub this \dt{rule determinacy}.  
One consequence is that if a configuration can step to fault then that is the only possible step.

Among these non-sequence rules,
\rn{bSplitL}, \rn{bSplitLX}, and \rn{bWhL} take left-only steps, leaving the right side unchanged;
whereas \rn{bSplitR}, \rn{bSplitRX}, and \rn{bWhR} take right-only steps.
All the other rules are for both-sides steps.


\begin{lemma}[bi-to-unary correspondence]\label{lem:bi-to-unary-plus}
Suppose $\phi$ is a $\Phi$-interpretation, with only unary specifications.
(a) For any step
$\configr{BB}{\sigma}{\sigma'}{\mu}{\mu'}\biTrans{}{\phi} \configr{CC}{\tau}{\tau'}{\nu}{\nu'}$, we have either
$\configm{\Left{BB}}{\sigma}{\mu} = \configm{\Left{CC}}{\tau}{\nu}$ or
$\configm{\Left{BB}}{\sigma}{\mu}\trans{\phi}\configm{\Left{CC}}{\tau}{\nu}$.
\emph{Mutatis mutandis} for the right.
(b) For any trace $T$ 
via $\biTrans{}{\phi}$, 
there are unique traces $U,V$ via $\trans{\phi}$
and $l,r$ such that $align(l,r,T,U,V)$.
\end{lemma}
\begin{proof}
Part (a) is by case analysis of the biprogram transition rules.
Rules \rn{bCall} and \rn{bCallX} are not relevant because they are for relational specifications and $\phi$ has only unary ones.
In all other cases, it is straightforward to check that the rule corresponds to a unary step on one or both sides, 
and in case it is a step on just one side the other side remains unchanged.

For part (b) the proof goes by induction on $T$ and case analysis on the rule by which the last step was taken.
Recall that traces are indexed from 0.  
The base case is $T$ comprised of a single configuration, $T_0$.
Let $U$ be $\Left{T_0}$, $V$ be $\Right{T_0}$, and let both $l$ and $r$ be the singleton mapping $\{ (0,0) \}$.
For the induction step, suppose $T$ has length $n+1$ and let $S$ be the prefix including all but the last configuration $T_n$.  
By induction hypothesis we get $l,r,U,V$ such that $align(l,r,S,U,V)$.
There are three sub-cases, depending on whether the step from $T_{n-1}$ to $T_n$ is a
left-only step (rule \rn{bSplitL} or \rn{bWhL}), or right-only, or both sides.
In the case of left-only, Let $U'$ be $U \Left{T_n}$, let $l'$ be $l\union \{ (n,len(U)) \}$.
Then $align(l',r,T,U',V)$.  
The other two sub-cases are similar. 
\end{proof}

Next, we need a result going from unary to biprogram traces, which is more intricate.

\begin{lemma}[unary-to-bi correspondence]\label{lem:unary-to-bi}
Suppose $\Phi$ has only unary specifications, and $\phi$ is a $\Phi$-interpretation.
Let $cfg$ be a biprogram configuration.
Let $U$ be a trace via $\phi$ from $\Left{cfg}$, 
and $V$ via $\phi$ from $\Right{cfg}$.
Then there is trace $T$ via $\phi$ from $cfg$
and traces $W$ from $\Left{cfg}$ and
$X$ from $\Right{cfg}$
and $l,r$ with $align(l,r,T,W,X)$, such that either
\begin{list}{}{}
\item[(a)] $U\leq W$ and $V\leq X$,
\item[(b)] $U\leq W$ and $W$ faults next 
and so does $T$, 
\item[(c)] $V\leq X$ and $X$ faults next 
and so does $T$,
\item[(d)] $U\leq W$ and $W$ is diverging and so is $T$,
\item[(e)] $V\leq X$ and $X$ is diverging and so is  $T$, or
\item[(f)] $W\leq U$ or $X\leq V$ and the last configuration of $T$
(which is aligned with the last ones of $W$ and $X$) faults due to 
lack of agreement for if-biprogram or while-biprogram, i.e., transitions \rn{bWhX}, \rn{bIfX}.
\end{list}
Moreover, if $U,V$ are the projections of a biprogram trace (see Lemma~\ref{lem:bi-to-unary-plus}) then case (f) does not occur.
\end{lemma}
\begin{proof}
%
First, we define an iterative procedure in which 
$l,r,W,X,T$ are treated as mutable variables.  It maintains this invariant: 
\[ align(l,r,T,W,X) \mbox{ and } (U\leq W \lor  W\leq U) \mbox{ and } (V\leq X \lor X\leq V)\]
$\bullet$ Initialize $W,X,T$ to the singleton traces $\Left{cfg}$, $\Right{cfg}$, and $cfg$ respectively.
Let $l(0)=0$ and $r(0)=0$ (and otherwise $l$ and $r$ are undefined).  \\
$\bullet$ While $(U\nleq W \mbox{ or } V\nleq X)$ and neither $W$, $X$, nor $T$ faults next, do the following:
\begin{list}{}{}
\item[(step A)] Let $k=len(T)-1$, so $k$ is the index of the last configuration of $T$.
Note that $l(k)$ and $r(k)$ index the last configurations of $U$ and $W$ respectively.
\item[(step B)]
If the rule that applies to $T_k$ is a left-only step (rule \rn{tSplitL} or \rn{tWhL}, since $T$ does not fault next),
then extend $l$ by $l(k+1)=l(k)+1$ (noting this is less than $len(U)$ because the loop guard and invariant imply $W$ is a strict prefix of $U$) and extend $r$ by $r(k+1)=r(k)$.
If right-only, extend by $r(k+1)=r(k)+1$ and $l(k+1)=l(k)$.
Otherwise, extend $l(k+1)=l(k)+1$ and $r(k+1)=r(k)+1$.
\item[(step C)]
Extend $T$ by one step via $\biTrans{}{\phi}$.
There may be nondeterministic choices to make on one or both sides, due to allocation and due to context calls; resolve these choices to match the configurations $U_{l(k+1)}$ and/or $V_{r(k+1)}$.
For allocation, this can be done because the same allocator is used by $\trans{\phi}$ and $\biTrans{}{\phi}$.
For context call, this can be done because the same interpretation is used.
(Because $\phi$ has only unary specifications, context calls go by rule \rn{bSync},
\rn{bSplitL}, or \rn{bSplitR}, not \rn{bCall}.)

\end{list}
To see that the invariant holds following step C, note that 
the invariant implies $\Left{T_k} = W_{l(k)}$ and $\Right{T_k} = X_{r(k)}$, for $k=len(T)-1$.
Then by construction 
we get a match for the new configuration: $\Left{T_{k+1}} = W_{l(k+1)}$ and $\Right{T_{k+1}} = X_{r(k+1)}$.

Now we can prove the lemma. If the loop terminates because guard condition 
$(U\nleq W \lor V\nleq X)$ became false then we have (a).
If it terminates because $W$ faults next then we have (b),
using invariant $(U\leq W\lor W\leq U)$.
Similarly, we get (c) if it terminates because $X$ faults next.
If it terminates because $T$ faults, and case (a) does not hold, then we have (f) owing to the invariants $U\leq W \lor  W\leq U$ and $V\leq X \lor X\leq V$.
If the loop fails to terminate we get (d) or (e). The reason is that every iteration lengthens both $T$ and either $W$ or $X$ (or both), so eventually either $U\leq W$ or $V\leq X$.  It is possible for left-only steps to diverge while
$V$ is still a proper prefix of $X$, as shown in earlier examples, and then we have (d); \emph{mutatis mutandis} for (e).

Finally, suppose $U,V$ are the projections of a biprogram trace $S$ (see Lemma~\ref{lem:bi-to-unary-plus}).
Then case (f) cannot happen: $T$ cannot fault until at least $S\leq T$, at which point the loop terminates and case (a) applies.
\end{proof}

\paragraph*{Proof Theorem~\ref{thm:biprogram-soundness}.}

Given terminated traces $U$ and $V$ of $\configm{C}{\sigma}{\mu}$ and $\configm{C'}{\sigma'}{\mu'}$, we can apply Lemma~\ref{lem:unary-to-bi} to obtain a trace $T$ of $\configr{\splitbi{C}{C'}}{\sigma}{\sigma'}{\mu}{\mu'}$
for which condition (a) in the Lemma holds. 
By correspondence (Lemma~\ref{lem:bi-to-unary-plus}(a)),
 $T$ must be terminated, so by the correctness judgment we get $\Q$ in the final state.
To prove safety, suppose $\configm{C}{\sigma}{\mu}$ can fault.  
Then by semantics, $\configr{\splitbi{C}{C'}}{\sigma}{\sigma'}{\mu}{\mu'}$
takes left-only steps until it reaches the fault, contrary to the assumed correctness judgment.
Suppose $\configm{C'}{\sigma'}{\mu'}$ can fault.  
By semantics, and the assumption that $C$ does not diverge from $\sigma$, $\configr{\splitbi{C}{C'}}{\sigma}{\sigma'}{\mu}{\mu'}$
reaches $\configr{\splitbi{\skipc}{C'}}{\tau}{\sigma'}{\mu}{\mu'}$ for some $\tau$, and then proceeds with right-only steps for $C'$.
So a fault of $C'$ gives rise to a fault of $\splitbi{C}{C'}$, contrary to the assumed correctness judgment.

\section{Soundness theorem}\label{app:sound}

\begin{theorem}\label{thm:sound}
All the proof rules for relational correctness judgments are sound with respect to the semantics in Sec.~\ref{app:relProofRules}.
\end{theorem}

The soundness proofs are straightforward for many of the rules.
As noted in the body of the paper, the proof for $\rn{rLink}$
(Fig.~\ref{fig:proofrulesR}) 
follows the lines of the soundness proof for the linking rules in~\cite{RegLogJrnII} and~\cite{BanerjeeNN15}.
It involves induction on biprogram traces, and the relational hypothesis can be used because the relevant context 
calls are aligned.

\paragraph*{Soundness of rule \rn{rWeave}}

We prove this rule in detail.  The argument may illuminate some design choices in the semantics of biprograms.

We write $\equiv$ for syntactic identity.

\begin{lemma}[weave and project]\label{lem:weave-project}
If $CC\weave DD$ then
$\Left{CC}\equiv \Left{DD}$ and 
$\Right{CC}\equiv \Right{DD}$.
\end{lemma}
\begin{proof}
By induction on the rules for $\weave$ (near the end of Sec.~\ref{app:relProofRules}),
making straightforward use of the definitions of the syntactic projections.

As an example, 
for the if-else axiom we have
$\Left{ \Splitbi{ \ifc{E}{C}{D} }{ \ifc{E'}{C'}{D'} } }
\equiv \ifc{E}{C}{D}
\equiv \ifc{E}{\Left{\splitbi{C}{C'}}}{\Left{\splitbi{D}{D'}}}
\equiv \Left{ \ifcbi{E\smallSplitSym E'}{\splitbi{C}{C'} }{ \splitbi{D}{D'} }}$.

As an example inductive case, 
for the rule from $BB\weave CC$ infer $BB;DD \weave CC;DD$, we have
$\Left{BB;DD} \equiv \Left{BB};\Left{DD} \equiv \Left{CC};\Left{DD} \equiv \Left{CC;DD}$
where the middle step is by induction hypothesis.
\end{proof}

\begin{lemma}[weave and trace]\label{lem:weave-trace}
Suppose $\Phi$ has only unary specifications, $BB \weave CC$,
and $\phi$ is a $\Phi$-interpretation.
If 
$\configm{BB}{\sigma|\sigma'}{\mu|\mu'}
 \biTranStar{}{\phi}
 \configm{\syncbi{\skipc}}{\tau|\tau'}{\mu|\mu'}$
then 
$\configm{CC}{\sigma|\sigma'}{\mu|\mu'}$ either faults, diverges,
or 
$\configm{CC}{\sigma|\sigma'}{\mu|\mu'}
 \biTranStar{}{\phi}
 \configm{\syncbi{\skipc}}{\tau|\tau'}{\mu|\mu'}$
(for any $\sigma,\sigma',\tau,\tau',\mu,\mu'$).
\end{lemma}
\begin{proof}
Consider any trace $T$ of $BB$ from $\sigma|\sigma'$.
Let $U,V$ be the corresponding unary traces, given by Lemma~\ref{lem:bi-to-unary-plus}(b).
In light of Lemma~\ref{lem:weave-project},
we can obtain a trace $\hat{T}$ from
$\configm{CC}{\sigma|\sigma'}{\_|\_}$ satisfying the conditions of 
Lemma~\ref{lem:unary-to-bi}.

Suppose in particular that $T$ is a trace from $\configm{BB}{\sigma|\sigma'}{\mu|\mu'}$
that terminates in $\configm{\syncbi{\skipc}}{\tau|\tau'}{\mu|\mu'}$.
Conditions (d) and (e) in Lemma~\ref{lem:unary-to-bi} imply that $\hat{T}$ diverges; then we are done.
Conditions (b), (c), and (f) all imply that $\configm{CC}{\sigma|\sigma'}{\mu|\mu'}$ faults; again we are done.
The remaining condition, (a), 
implies that $\hat{T}$ covers all the steps of $T$
and since $T$ is terminated, so is $\hat{T}$. 
The $align$ conditions of Lemmas~\ref{lem:bi-to-unary-plus} and~\ref{lem:unary-to-bi}
imply that the final states of $\hat{T}$ are $\tau,\tau'$.
\end{proof}

Now we can prove soundness of rule \rn{rWeave}.

\begin{proof}
Suppose the premise and side conditions hold:
\begin{itemize}
\item $\Phi \rHVflowtr{}{}{\P}{DD}{\Q}{\eff}$.
\item $CC \weave DD$
\item $\Phi$ has only unary specifications
\item $\Left{DD}$ terminates from any $\Left{\P}$-state, 
and $\Right{DD}$ terminates from any $\Right{\P}$-state.
\end{itemize}
To show the conclusion $\Phi \rHVflowtr{}{}{\P}{CC}{\Q}{\eff}$ 
(as per the semantics of judgments, in Sec.~\ref{app:relProofRules}),
consider any $\Phi$-interpretation $\phi$.
Consider any $\pi$ and any $\sigma,\sigma'$ such that $\sigma|\sigma'\models_\pi \P$.
Suppose $CC$ terminates from $\sigma|\sigma'$ in final states $\tau|\tau'$.
By Lemma~\ref{lem:weave-trace}, execution of $DD$ from $\sigma|\sigma'$ either faults, diverges, or terminates 
in $\tau|\tau'$.
It cannot fault, owing to the premise for $DD$.  
It cannot diverge:  Because if the traces of $DD$ from 
$\sigma|\sigma'$ could be extended without bound, then by Lemma~\ref{lem:bi-to-unary-plus}
either traces of $\Left{DD}$ from $\sigma$ could be extended without bound, 
or traces of $\Right{DD}$ from $\sigma'$ could be---which contradicts the termination conditions,
since $\sigma\models \Left{\P}$ and $\sigma'\models\Right{\P}$ (because
$\P\imp\Left{\P}\land\Right{\P}$ is valid).
So $DD$ terminates in $\tau|\tau'$.
Now conditions (Post), (Write Effect), and (Read Effect) for $CC$ are immediate from the premise for $DD$.

It remains to show that safety for $CC$ follows from safety for $DD$.  This is a direct consequence of a general property of weaving, which we state as Lemma~\ref{lem:weave-safe}.
\end{proof}

\begin{lemma}\label{lem:weave-safe}
Consider any $CC,DD$ such that $CC\weave DD$.
Consider any interpretation $\phi$ of some $\Phi$ with only unary specifications.\footnote{The specifications are irrelevant; all that matters is that the interpretation is unary.}
For any $\sigma,\sigma',\mu,\mu'$,
if $\configm{CC}{\sigma|\sigma'}{\mu|\mu'}$ can fault then
$\configm{DD}{\sigma|\sigma'}{\mu|\mu'}$ can fault or diverge.
\end{lemma}
\begin{proof}
By rule induction on the definition of $\weave$.
In each case, 
we assume the lhs ($CC$) faults and show that rhs ($DD$) either faults or diverges,
by an analysis using the biprogram semantics (Figs.~\ref{fig:biprogTrans} and~\ref{fig:biprogTransU}).

In reasoning about transitions that do not manipulate the environment we omit $\mu,\mu'$.
Also, we omit $\phi$ from $\trans{\phi}$ and $\biTrans{}{\phi}$.

The base cases are the weaving axioms.
\begin{description}
\item[case] 
$\splitbi{A}{A} \weave \syncbi{A}$

By semantics there are two ways $\splitbi{A}{A}$ can fault:
\begin{itemize}
\item $\config{\splitbi{A}{A}}{\sigma|\sigma'} \biTrans{}{} \Fault$
by transition \rn{bSplitLX}, where $\config{A}{\sigma}\trans{}{}\Fault$.
\item $\config{\splitbi{A}{A}}{\sigma|\sigma'} \biTrans{}{} 
       \config{\splitbi{\skipc}{A}}{\sigma|\sigma'} \biTrans{}{} \Fault$
by \rn{bSplitL} and then \rn{bSplitRX}, where $\config{A}{\sigma'}\trans{}{}\Fault$.
\end{itemize}
In either case we have 
$\config{\syncbi{A}}{\sigma|\sigma'} \biTrans{}{} \Fault$ by \rn{bSyncX}.

\item[case]
$\Splitbi{C;D}{C';D'} \weave \splitbi{C}{C'};\splitbi{D}{D'}$

There are four ways $\Splitbi{C;D}{C';D'}$ can fault:
\begin{enumerate}
\item $\config{\Splitbi{C;D}{C';D'}}{\sigma|\sigma'} \biTranStar{}{}
       \config{\Splitbi{C_0;D}{C';D'}}{\tau|\sigma'} \biTrans{}{} \Fault$ 
      for some $C_0,\tau$ such that $\config{C_0}{\tau}\trans{}{}\Fault$
(by some number of instances of \rn{bSplitL} and then \rn{bSplitLX}).

\item $\config{\Splitbi{C;D}{C';D'}}{\sigma|\sigma'} \biTranStar{}{}
       \config{\Splitbi{D}{C';D'}}{\upsilon|\sigma'} \biTranStar{}{}
       \config{\Splitbi{D_0}{C';D'}}{\tau|\sigma'} \biTrans{}{} \Fault$
       where $\config{D_0}{\tau}\trans{}\Fault$
(for some $D_0,\upsilon,\tau$, and again by \rn{bSplitL} and \rn{bSplitLX}).

\item $\config{\Splitbi{C;D}{C';D'}}{\sigma|\sigma'} \biTranStar{}{}
       \config{\Splitbi{D}{C';D'}}{\upsilon|\sigma'} \biTranStar{}{}
       \config{\Splitbi{\skipc}{C';D'}}{\tau|\sigma'} \biTranStar{}{}
       \config{\Splitbi{\skipc}{C'_0;D'}}{\tau|\tau'} \biTrans{}{} \Fault$
       where $\config{C'_0}{\tau'}\trans{}\Fault$
(by \rn{bSplitL}, \rn{bSplitR}, and \rn{bSplitRX}).

\item $\config{\Splitbi{C;D}{C';D'}}{\sigma|\sigma'} \biTranStar{}{}
       \config{\Splitbi{D}{C';D'}}{\upsilon|\sigma'} \biTranStar{}{}
       \config{\Splitbi{\skipc}{C';D'}}{\tau|\sigma'} \biTranStar{}{}
       \config{\Splitbi{\skipc}{D'}}{\tau|\upsilon'}  \biTranStar{}{}
       \config{\Splitbi{\skipc}{D'_0}}{\tau|\tau'}  \biTrans{}{} \Fault$
       where $\config{D'_0}{\tau'}\trans{}\Fault$
\end{enumerate}
For each case we show that the woven biprogram $\splitbi{C}{C'};\splitbi{D}{D'}$ faults or diverges.
\begin{enumerate}
\item $\config{\splitbi{C}{C'};\splitbi{D}{D'}}{\sigma|\sigma'} \biTranStar{}{}
       \config{\splitbi{C_0}{C'};\splitbi{D}{D'}}{\tau|\sigma'} \biTrans{}{} \Fault$ 
by \rn{bSplitL} and \rn{bSplitLX}, because $\config{C_0}{\tau}\trans{}{}\Fault$

\item We have $\config{\splitbi{C}{C'};\splitbi{D}{D'}}{\sigma|\sigma'} \biTranStar{}{}
       \config{\splitbi{\skipc}{C'};\splitbi{D}{D'}}{\upsilon|\sigma'}$.
From this point, $C'$ could fault, or diverge, in which case $\splitbi{C}{C'};\splitbi{D}{D'}$ faults, or diverges,  and we are done.
Otherwise, execution can continue as 
$\config{\splitbi{\skipc}{C'};\splitbi{D}{D'}}{\upsilon|\sigma'} \biTranStar{}{}
 \config{\splitbi{D}{D'}}{\upsilon|\upsilon'} \biTranStar{}{}
 \config{\splitbi{D_0}{D'}}{\tau|\upsilon'} \biTrans{}{} \Fault$
because $\config{D_0}{\tau}\trans{}\Fault$.

\item 
$\config{\splitbi{C}{C'};\splitbi{D}{D'}}{\sigma|\sigma'} \biTranStar{}{}
 \config{\splitbi{\skipc}{C'};\splitbi{D}{D'}}{\upsilon|\sigma'} \biTranStar{}{}
 \config{\splitbi{\skipc}{C'_0};\splitbi{D}{D'}}{\upsilon|\tau'} \biTrans{}{} \Fault$

\item 
$\config{\splitbi{C}{C'};\splitbi{D}{D'}}{\sigma|\sigma'} \biTranStar{}{}
 \config{\splitbi{\skipc}{C'};\splitbi{D}{D'}}{\upsilon|\sigma'} \biTranStar{}{}
 \config{\splitbi{D}{D'}}{\upsilon|\upsilon'} \biTranStar{}{}
 \config{\splitbi{\skipc}{D'}}{\tau|\upsilon'} \biTranStar{}{}
 \config{\splitbi{\skipc}{D'_0}}{\tau|\tau'} \biTrans{}{} \Fault$

\end{enumerate}

\item[case]
$\Splitbi{ \ifc{E}{C}{D} }{ \ifc{E'}{C'}{D'} } 
   \weave \ifcbi{E\smallSplitSym E'}{\splitbi{C}{C'} }{ \splitbi{D}{D'} }$

If $\sigma(E)\neq \sigma'(E')$ then 
$\config{\ifcbi{E\smallSplitSym E'}{\splitbi{C}{C'} }{ \splitbi{D}{D'} }}{\sigma|\sigma'}\biTrans{}{}\Fault$ by \rn{bIfX} and we are done.
Otherwise, $\sigma(E) = \sigma'(E')$, and we consider the four ways that the lhs can fault.
\begin{enumerate}
\item 
  $\config{ \Splitbi{ \ifc{E}{C}{D} }{ \ifc{E'}{C'}{D'} } }{\sigma|\sigma'} \biTrans{}{}
   \config{ \Splitbi{ C }{ \ifc{E'}{C'}{D'} } }{\sigma|\sigma'} \biTranStar{}{} 
   \config{ \Splitbi{ C_0 }{ \ifc{E'}{C'}{D'} } }{\tau|\sigma'} \biTrans{}{} \Fault$
by \rn{bSplitL} and then \rn{bSplitLX}, where $\sigma(E)=true$.

Then for the rhs we have 
$\config{ \ifcbi{E\smallSplitSym E'}{\splitbi{C}{C'} }{ \splitbi{D}{D'} }}{\sigma|\sigma'} \biTrans{}{}
 \config{ \splitbi{C}{C'} }{\sigma|\sigma'} \biTranStar{}{}
 \config{ \splitbi{C_0}{C'} }{\tau|\sigma'} \biTrans{}{} \Fault$

\item 
  $\config{ \Splitbi{ \ifc{E}{C}{D} }{ \ifc{E'}{C'}{D'} } }{\sigma|\sigma'} \biTrans{}{}
   \config{ \Splitbi{ C }{ \ifc{E'}{C'}{D'} } }{\sigma|\sigma'} \biTranStar{}{} 
   \config{ \Splitbi{ \skipc }{ \ifc{E'}{C'}{D'} } }{\tau|\sigma'} \biTrans{}{} 
   \config{ \Splitbi{ \skipc }{ C'}}{\tau|\sigma'} \biTranStar{}{} 
   \config{ \Splitbi{ \skipc }{ C'_0}}{\tau|\tau'} \biTrans{}{} \Fault$
where $\sigma(E)=true=\sigma'(E')$.

Then for the rhs we have
$\config{ \ifcbi{E\smallSplitSym E'}{\splitbi{C}{C'} }{ \splitbi{D}{D'} }}{\sigma|\sigma'} \biTrans{}{}
\config{ \splitbi{C}{C'} }{\sigma|\sigma'} \biTranStar{}{}
\config{ \splitbi{\skipc}{C'} }{\tau|\sigma'} \biTranStar{}{}
\config{ \splitbi{\skipc}{C'_0} }{\tau|\tau'} \biTrans{}{} \Fault$
\end{enumerate}
The other two cases are symmetric.

\item[case]
$\Splitbi{ \letcom{m}{B}{C} }{ \letcom{m}{B'}{C'} }
    \weave \letcombi{m}{\splitbi{B}{B'}}{ \splitbi{C}{C'} } $

Similar to the preceding cases, using that transitions taken on one side are not affected by the presence or absence of 
a binding for $m$ on the other side.  

\item[case]
$\Splitbi{ \whilec{E}{C} }{ \whilec{E'}{C'} }
   \weave \whilecbiA{E\smallSplitSym E'}{\P\smallSplitSym \P'}{\splitbi{C}{C'}}$

There are two ways the lhs can fault.
\begin{enumerate}
\item 
$\config{ \Splitbi{ \whilec{E}{C} }{ \whilec{E'}{C'} } }{\sigma|\sigma'} \biTranStar{}{}
\config{ \Splitbi{ C;\whilec{E}{C} }{ \whilec{E'}{C'} } }{\upsilon|\sigma'} \biTranStar{}{}
\config{ \Splitbi{ C_0;\whilec{E}{C} }{ \whilec{E'}{C'} } }{\tau|\sigma'} \biTrans{}{} \Fault$
where $\upsilon$ is the state after the $n$th completed iteration on the left, and 
$\config{C_0}{\tau}\trans{}\Fault$.

\item
$\config{ \Splitbi{ \whilec{E}{C} }{ \whilec{E'}{C'} } }{\sigma|\sigma'} \biTranStar{}{}
\config{ \Splitbi{ \skipc }{ \whilec{E'}{C'} } }{\tau|\sigma'} \biTranStar{}{}
\config{ \Splitbi{ \skipc }{ C';\whilec{E'}{C'} } }{\tau|\upsilon'} \biTranStar{}{}
\config{ \Splitbi{ \skipc }{ C'_0;\whilec{E'}{C'} } }{\tau|\tau'} \biTrans{}{} \Fault$
where $\upsilon'$ is the state after the $k$th completed iteration on the right, and
$\config{C'_0}{\tau'}\trans{}\Fault$.
\end{enumerate}

Note that $\Left{\splitbi{C}{C'}}$ is $C$ and $\Right{\splitbi{C}{C'}}$ is $C'$.
Consider a trace from
$\config{ \whilecbiA{E\smallSplitSym E'}{\P\smallSplitSym \P'}{\splitbi{C}{C'}} }{\sigma|\sigma'}$.
If it faults due to \rn{rWhX} we are done.
Otherwise it can be segmented into the iterates, each of which begins with 
a step by \rn{rWhL}, \rn{rWhR}, or \rn{rWhTT}, and accordingly executes 
$\splitbi{C}{\skipc}$, $\splitbi{\skipc}{C'}$, or $\splitbi{C}{C'}$.
These can be put in correspondence with some or all the iterates of 
$\Splitbi{ \whilec{E}{C} }{ \whilec{E'}{C'} } $ 
so that the same states are reached.
 
Whereas the lhs never executes $C'$ unless and until $C$ terminates 
(and does not fault), the rhs may do some iterations of $C'$ before all iterations of $C$ have been done.
Hence, if some iteration of $C'$ diverges, then the rhs diverges and we are done.
In the absence of divergence, the rhs eventually reaches either the $n$th iteration on the left (case 1 above)
or the $k$th iteration on the right (case 2 above).
From that point, either steps by \rn{bSplitL} lead to a point where 
we get a fault by \rn{bSplitLX} ($\config{C_0}{\tau}\trans{}\Fault$) or
steps by \rn{bSplitR} lead to a point where 
we get a fault by \rn{bSplitRX} ($\config{C'_0}{\tau'}\trans{}\Fault$).
\end{description}

Having dispensed with the base cases, we turn to the 
inductive cases which each have as premise that 
$BB\weave CC$.  The inductive hypothesis:   
for any $\sigma,\sigma',\mu,\mu'$,
if $\configm{BB}{\sigma|\sigma'}{\mu|\mu'}$ can fault then
$\configm{CC}{\sigma|\sigma'}{\mu|\mu'}$ can fault or diverge.
\begin{description}
\item[case]
$BB;DD \weave CC;DD$

There are two ways the lhs can fault.
\begin{enumerate}
\item $BB;DD$ faults from $\sigma,\sigma'$ 
because $BB$ does.

Then by induction hypothesis, $CC$ (and hence $CC;DD$) diverges or faults.

\item 
$\config{BB;DD}{\sigma|\sigma'} \biTranStar{}{}
 \config{DD}{\upsilon|\upsilon'} \biTranStar{}{}
 \config{DD_0}{\tau|\tau'} \biTrans{}{} \Fault$,
with $\config{BB}{\sigma|\sigma'} \biTranStar{}{} \config{\syncbi{\skipc}}{\upsilon|\upsilon'}$.

Then by Lem.~\ref{lem:weave-trace},
$\config{CC}{\sigma|\sigma'}$ either faults or diverges (and then we are done) or
it terminates in $\upsilon,\upsilon'$.
In the latter case we have
$\config{CC;DD}{\sigma|\sigma'} \biTranStar{}{}
 \config{DD}{\upsilon|\upsilon'} \biTranStar{}{}
 \config{DD_0}{\tau|\tau'} \biTrans{}{} \Fault$.
\end{enumerate}

\item[case]
$DD;BB \weave DD;CC$

If lhs faults in $DD$ then so does rhs.  
Otherwise both lhs and rhs reach the same intermediate states upon termination of $DD$, 
from which $BB$ faults.  So by induction hypothesis, $CC$ faults from those states.

\item[case]
$\ifcbi{E\smallSplitSym E'}{BB}{DD} \weave \ifcbi{E\smallSplitSym E'}{CC}{DD}$

If $\sigma(E)\neq\sigma'(E)$ then lhs and rhs both fault immediately.

If $\sigma(E)= false = \sigma'(E)$ then both sides take one step to the same configuration
$\config{ DD }{\sigma|\sigma'}$ so fault on lhs implies fault on rhs.

If $\sigma(E)= true = \sigma'(E)$ then 
fault on the lhs looks like 
$\config{ \ifcbi{E\smallSplitSym E'}{BB}{DD} }{\sigma|\sigma'} \biTrans{}{}
\config{ BB }{\sigma|\sigma'} \biTranStar{}{} \Fault$ 
so for the rhs it suffices to apply the induction hypothesis.

\item[case]
$\ifcbi{E\smallSplitSym E'}{DD}{BB} \weave \ifcbi{E\smallSplitSym E'}{DD}{CC}$

Symmetric to the preceding case.

\item[case]
$\letcombi{m}{\splitbi{B}{B'}}{ BB } \weave \letcombi{m}{\splitbi{B}{B'}}{ CC }$

A faulting trace for lhs has the form
$\configm{ \letcombi{m}{\splitbi{B}{B'}}{ BB } }{\sigma|\sigma'}{\mu|\mu'} \biTrans{}{}
 \configm{ BB }{\sigma|\sigma'}{\dot{\mu}|\dot{\mu'}} \biTranStar{}{}
 \configm{ BB_0 }{\tau|\tau'}{\dot{\mu}|\dot{\mu'}} \biTrans{}{} \Fault$
where $\dot{\mu}$ is $\extend{\mu}{m}{B}$ and $\dot{\mu}'$ is $\extend{\mu'}{m}{B'}$.
By the induction hypothesis (which is stated for all environments),
we get a fault or divergence for the rhs.

\item[case]
$\whilecbiA{E\smallSplitSym E'}{\P\smallSplitSym \P'}{BB} \weave
  \whilecbiA{E\smallSplitSym E'}{\P\smallSplitSym \P'}{CC}$

A faulting trace $T$ of $\whilecbiA{E\smallSplitSym E'}{\P\smallSplitSym \P'}{BB}$ from $\sigma,\sigma'$
can be segmented into $n$ completed iterates (each executing $BB$, $\splitbi{\Left{BB}}{\skipc}$, or $\splitbi{\skipc}{\Right{BB}}$),
followed by a partial iterate that faults.

Claim: a trace of $ \whilecbiA{E\smallSplitSym E'}{\P\smallSplitSym \P'}{CC}$ from $\sigma,\sigma'$
begins with $k\leq n$ completed iterates,
ending in the same states as the corresponding iteration in $T$,
and executing $\splitbi{\Left{CC}}{\skipc}$, $\splitbi{\skipc}{\Right{CC}}$, or $CC$,
according to whether the corresponding iteration in $T$ is left, right, or both.
These $k$ completed iterates are possibly followed by fault or divergence.

Proof of claim: by induction the iterates in $T$.  For the induction step, having the same states implies
the rhs takes either a left-, right-, or both-iteration just as the lhs did (because the weaving did not change the alignment guards).  If it is a one-sided iteration, by Lemma~\ref{lem:weave-project} the rhs is executing the same command, hence it faults.  If it is a both-sided iteration, i.e., a terminated execution of $BB$,
then Lemma~\ref{lem:weave-trace} tells us that $CC$ either faults, diverges,
or terminates in the same states.  The claim is proved.

Returning to the partial iterate of $BB$ that faults,
if it is one-sided, say $\Left{BB}$, then by Lemma~\ref{lem:weave-project} we have $\Left{CC}\equiv \Left{BB}$
and hence rhs faults.
If it is a both-sided iteration (i.e., beginning with \rn{bWhTT}),
the induction hypothesis applies, to yield a fault or divergence of $CC$.
\end{description}
\end{proof}


There is an obvious rule like \rn{rWeave} but using the transitive closure $\weave^*$.
This is admissible: any proof using that rule can be transformed to one making repeated use of \rn{rWeave}.
Indeed, it is sound: soundness can be proved for $\weave^n$, by induction on $n$,
by using the above argument for each $\weave$ in sequence.

\paragraph*{Framed reads and rules $\rn{rSeq}$ and $\rn{rWh}$}

One complication in the rules for sequence and loops (\rn{rSeq} and \rn{rWh}) 
is already present in the unary rules for sequence and loops.
The issue is that, because frame conditions can be expressed in terms of mutable locations (often ghost variables and fields), 
sound sequencing of judgments must avoid interference with those locations---so that the interpretation of an effect can be preserved over a command.
Soundness is achieved by the immunity conditions.  These are already present in the unary logic~\cite{RegLogJrnI}, but use of immunity is more delicate with the addition of read effects in~\cite{BanerjeeNN15}.  A key restriction is that specifications have framed reads (see Sec.~\ref{app:unary}).  This ensures that read effects are preserved under suitable immunity conditions. 

To state the key lemma, first
we define the image of refperm $\pi$ on arbitrary location set $W$.
This is written $\pi(W)$ and defined for variables by
$ x\in \pi(W) \mbox{ iff } x \in W$ and for heap locations by
$o.f\in \pi(W) \mbox{ iff } (\pi^{-1}(o)).f \in W $.
Second, we note a basic lemma:
Suppose $\eff$ has framed reads.
If $\agree(\sigma,\sigma',\eff,\pi)$ then
$\rlocs(\sigma',\eff)=\pi(\rlocs(\sigma,\eff))$ and hence
$\agree(\sigma',\sigma,\eff,\pi^{-1})$.

The key lemma is a bit technical.  
Keep in mind that the read effect part of a correctness judgment, and also
the read effect condition for context interpretations, quantifies over all pairs of runs.
In particular, for any two initial states $\sigma,\sigma'$, the conditions will be instantiated both with
$\sigma$ on the left and with $\sigma$ on the right.

\begin{lemma}[preservation of agreement]\label{lem:selfframing-agreement2}
\upshape
Suppose $\eff$ has framed reads.
Suppose $\sigma,\sigma'\allowDep\tau,\tau' \models \eff$ and 
$\sigma',\sigma\allowDep\tau',\tau \models \eff$.
Suppose  $\agree(\sigma, \sigma',\eff,\pi)$ and 
$\agree(\sigma', \sigma,\eff,\pi^{-1})$.
Let $\rho$ be any refperm $\rho\supseteq\pi$
for which 
\(\Lagree(\tau,\tau',\rho,\freshLocs(\sigma,\tau)\union \written(\sigma,\tau))
\).
Then for any set of locations $W$ in $\sigma$, if $\Lagree(\sigma,\sigma',\pi,W)$ then 
$\Lagree(\tau,\tau',\rho,W)$.
\end{lemma}
Existence of $\rho$ for 
$\Lagree(\tau,\tau',\rho,\freshLocs(\sigma,\tau)\union \written(\sigma,\tau))$
is a consequence of \mbox{$\sigma,\sigma'\allowDep\tau,\tau' \models \eff$}.
The result says that agreement on an arbitrary set $W$ is preserved.
For a proof, see \cite{BanerjeeNN15}.
